\documentclass[sigconf]{aamas}

\usepackage{balance} %

\usepackage{multirow,array}

\newcolumntype{L}[1]{>{\raggedright\let\newline\\\arraybackslash\hspace{0pt}}m{#1}}
\newcolumntype{C}[1]{>{\centering\let\newline\\\arraybackslash\hspace{0pt}}m{#1}}
\newcolumntype{R}[1]{>{\raggedleft\let\newline\\\arraybackslash\hspace{0pt}}m{#1}}

\setcopyright{ifaamas}
\acmConference[AAMAS '24]{Proc.\@ of the 23rd International Conference
on Autonomous Agents and Multiagent Systems (AAMAS 2024)}{May 6 -- 10, 2024}
{Auckland, New Zealand}{N.~Alechina, V.~Dignum, M.~Dastani, J.S.~Sichman (eds.)}
\copyrightyear{2024}
\acmYear{2024}
\acmDOI{}
\acmPrice{}
\acmISBN{}

\acmSubmissionID{991}

\title{Computing Optimal Commitments to Strategies and Outcome-Conditional Utility Transfers}

\author {Nathaniel Sauerberg}
\affiliation{
  \institution{University of Texas--Austin}
  \city{Austin}
  \country{United States}}
\email{njs@cs.utexas.edu}

\author{Caspar Oesterheld}
\affiliation{
  \institution{Carnegie Mellon University}
  \city{Pittsburgh}
  \country{United States}}
\email{oesterheld@cmu.edu}

\begin{abstract}
Prior work \cite{conitzer_optimal_strategy_to_commit,commitment-to-correlated-strats} has studied the computational complexity of computing optimal strategies to commit to in Stackelberg or leadership games, where a leader commits to a strategy which is then observed by one or more followers. 
We extend this setting to one in which the leader can additionally commit to outcome-conditional utility transfers.
In this setting, we characterize the computational complexity of finding optimal commitments for normal-form and Bayesian games. 
We find a mix of polynomial time algorithms and \NP-hardness results.
Then, we allow the leader to additionally commit to a signaling scheme based on her action, inducing a correlated equilibrium.
In this variant, optimal commitments can be computed efficiently for arbitrarily many players.

\end{abstract}

\keywords{Stackelberg Games;
Leadership Games;
Commitment Games;
Computational Complexity;
Transferable Utility;
NP-Hardness;
Linear Programming;
Bayesian Games
}

\newcommand{\BibTeX}{\rm B\kern-.05em{\sc i\kern-.025em b}\kern-.08em\TeX}

\usepackage{graphicx}
\usepackage{xcolor}

\usepackage{amsfonts, amsmath, amssymb, amsthm}
\usepackage{mathtools}
\usepackage{nicefrac}
\usepackage{enumitem}

\usepackage{thm-restate}

\theoremstyle{definition}

\newtheorem*{definition*}{Definition}

\newtheorem{theorem}{Theorem}[section]
\newtheorem*{theorem*}{Theorem}
\newtheorem{corollary}{Corollary}[theorem]
\newtheorem{lemma}[theorem]{Lemma}

\newtheorem{problem}{Problem}[section]
\newtheorem*{problem*}{Problem}
\newtheorem{proposition}{Proposition}[section]

\newenvironment{proofsketch}{%
  \renewcommand{\proofname}{Proof Sketch}\proof}{\endproof}

\usepackage[nameinlink]{cleveref}
\crefname{problem}{problem}{problems} 
\Crefname{problem}{Problem}{Problems} 

\newcommand{\E}{\mathop{\mathbb{E}}}
\newcommand{\eps}{\varepsilon}

\newcommand{\NP}[0]{$\mathsf{NP}$}
\newcommand{\NPH}[0]{\NP-hard}
\newcommand{\NPC}[0]{\NP-complete}
\newcommand{\BPTIME}[0]{\mathsf{BPTIME}}

\newcommand{\NE}[0]{\mathrm{NE}}
\newcommand{\Rev}[0]{\mathrm{Rev}}
\newcommand{\BR}[0]{\mathrm{BR}}

\DeclareMathOperator\supp{supp}

\usepackage{multirow}

\usepackage{tikz}

\DeclarePairedDelimiter\abs{\lvert}{\rvert}%
\DeclarePairedDelimiter\norm{\lVert}{\rVert}%
\makeatletter
\let\oldabs\abs
\def\abs{\@ifstar{\oldabs}{\oldabs*}}
\let\oldnorm\norm
\def\norm{\@ifstar{\oldnorm}{\oldnorm*}}
\makeatother

\usepackage{ifthen}
\usepackage{xcolor}

\usepackage{csquotes}%

\usepackage{float}
\floatstyle{plaintop}
\restylefloat{table}

\begin{document}

\pagestyle{fancy}
\fancyhead{}

\maketitle

\section{Introduction}

Since the early days of the field, game theorists have studied models of commitments and sequential play. 
Perhaps the first was \citet{Stackelberg1934}, who introduced his now-eponymous model to argue that a firm with the ability to commit to an action (a quantity of production) benefits from doing so in the duopoly competition model of \citet{cournot1838recherches}.
Commitment to pure and mixed strategies were invoked by \citet{von-VNM-zero-sum} to motivate the min-max and max-min values of zero-sum games. 
\citet{vonStengel-mixed-commitment} study the payoff implications of commitment to mixed strategies, comparing the resulting payoffs to those in the Nash and correlated equilibria in the simultaneous versions of the game. 
The first to study the computational complexity of computing optimal commitments were \citet{conitzer_optimal_strategy_to_commit}, who characterized the complexity of computing optimal pure and mixed strategies to commit to in normal-form and Bayesian games.
\citet{commitment-to-correlated-strats} introduced the notion of commitment to correlated strategies and showed that finding optimal such commitments is tractable even with arbitrarily many players.

More recently, a number of authors in game theory and multi-agent learning have studied the use of voluntary utility transfers and commitments to pay money \cite{rational-explanation-bribery,Haupt2023,Willis2023,yi2022learning,kolle2023learning,Wang2021, lupu2020gifting,k-implementation, complexity-and-algs-k-implementation,internal_implementation}. We are specifically interested in outcome-conditional payments, such as committing to pay a particular agent if they take a particular action. We see two primary motivations to study such commitments to payments. The first is descriptive: Commitments to pay others are a common feature of human economies. For example, companies pay employees for performing tasks. The second is normative: It may be good if commitments to payments were available to and used by agents, as commitments to payments often allow for cooperative equilibria.
For example, in the Prisoner’s Dilemma, one player may pay the other to cooperate \cite{Haupt2023,Willis2023}.

In this paper, we combine these two lines of work and study games in which players can commit to both strategies and outcome-conditional utility transfers.
To our knowledge, only \citet{rational-explanation-bribery} have considered such a setting before. (See \Cref{subsec:related-work} for a discussion of their contributions.)%

In the most basic case, we take a given normal-form game (the \enquote{base game}) and imagine that Player 1 (the \enquote{leader}) 
makes a commitment as follows:
As in Stackelberg games, she commits to taking a specific action or mixture over actions in the base game.
In addition, she can commit to transfer utility to other players based on the outcome of the game. 
Specifically, she commits to a payment function $P: A \rightarrow \mathbb{R}^{n-1}_{\geq 0}$, which we index $i=2 \dots n$, where $P_i(a)$ is the amount of utility the leader will transfer to Player $i$ if outcome $a$ obtains.%

\begin{table}[b]
\begin{tabular}{cccc}
                            & Left                      & Middle                          & Right                     \\ \cline{2-4} 
\multicolumn{1}{c|}{Top}    & \multicolumn{1}{c|}{$-1$, $0$} & \multicolumn{1}{c|}{$-1$, $-2$} & \multicolumn{1}{c|}{$-1$, $0$} \\ \cline{2-4} 
\multicolumn{1}{c|}{Bottom} & \multicolumn{1}{c|}{$\hphantom{-}2$, $0$} & \multicolumn{1}{c|}{$0$, $2$}       & \multicolumn{1}{c|}{$\hphantom{-}0$, $1$}  \\ \cline{2-4} 
\end{tabular}%
\caption{A game in which commitment to strategies and payments is beneficial to the leader.}
\label{ex:intro}
\end{table}

For example, consider the game in \Cref{ex:intro}.
Since Bottom is a strictly dominant strategy for Player 1, the only Nash equilibrium is (Bottom, Middle), resulting in payoffs of $(0, 2)$.
However, in our setting, the leader can commit to play the mixture $(\nicefrac{1}{3}, \nicefrac{2}{3})$ over (Top, Bottom) and to the payment function $P$ where $P(B, L) = 1$ and $P(a) = 0$ for all $a \neq (B, L)$, essentially promising to transfer $1$ utility to the follower if (Bottom, Left) is played.

Then, the follower's expected utility for playing Right is %
$\nicefrac{2}{3}$%
, as is his utility for playing Left since the their payoffs after payments are the same.
His utility for playing Middle is $2\cdot{}\nicefrac{2}{3} -2\cdot{}\nicefrac{1}{3} = \nicefrac{2}{3}$. 
Hence, the follower is indifferent between each of his actions and tie-breaks in favor of the leader, playing Left for a leader payoff of $-1\cdot{}\nicefrac{1}{3}+(2-1)\cdot{}\nicefrac{2}{3} = \nicefrac{1}{3}$.
(Of course, the leader could also pay slightly more to strictly incentivize Left.)

It turns out that both commitment to actions and payments are necessary for the leader to achieve positive utility in this game.
Intuitively, the leader cannot incentivize Left over Right without payments and cannot incentivize Left over Middle without commitment to (mixtures over) actions. 
For a more detailed analysis, see %
\Cref{lemma:analysis-intro-ex}. 

Settings involving commitment to both strategies and payments are ubiquitous in human interactions.
For example, governments often pass laws (commitments) that determine both systems of subsidies (payments) as well as government policies that influence the structure of economic interactions (actions). These actions include what the government buys and sells, what services and entitlements it provides to citizens, what infrastructure it maintains, and how it regulates sectors. 

Another example is companies' interactions with their employees. 
Here, companies set both general policies like rules, priorities, company hierarchy, etc. (commitment to strategies) and compensation policies like pay scales and bonus structure (commitment to payments), to which employees then respond. 
We might further imagine that leaders of subdivisions respond to the company’s commitments by in turn setting general and compensation policies within their subdivisions, thus resulting in a multi-step game where players make commitments in sequence. 

Both pure and mixed commitment are quite natural and common. In some cases, the followers observe the specific base game action the leader takes before making their decisions. In such cases, even if the leader does randomize, the followers would observe the specific draw from the distribution before making their choices, and so such randomization would not be helpful. 
In other cases, the followers can observe the leader's long-run distribution over actions before choosing their own actions, but cannot observe the leader's realized action in their specific instance.

\subsection{Contributions}
In this paper, we ask: How can we compute optimal commitments in such settings?
The answer depends on various features of the setting, such as the number of players, whether only the first or also later players can commit, which players (if any) have Bayesian types and whether correlation devices are available. We give an overview of our results here. 

In \Cref{sec:NE}, we study a setting without private information and without correlation devices. (See \Cref{table:NE-results} for a summary of the following results.) 
\begin{itemize}
    \item We show that in two-player games, optimal pure commitments can be found efficiently with dynamic programming
    (\Cref{Thm:2-player-pure-commitment}) and optimal mixed commitments can be found efficiently with linear programming (LP) (\Cref{Thm:2-player-mixed-strat})
    \item We show that if there are more than two players and only the first commits, the optimal commitment is NP-complete to find (\Cref{thm:3-player-single-commitment-hardness}).
    \item We show that in the case of three players in which the players commit in sequence, the pure commitment case can be solved efficiently with LP (\Cref{thm:3-player-iterated-pure-LP}) while the mixed commitment case is NP-hard (\Cref{theorem:3-player-iterative-mixed-commitment-hard}).
\end{itemize}

In \Cref{sec:CE}, we study a setting without private information but where correlation devices are available. 
Specifically, the leader can construct an arbitrary signaling scheme (as in \cite{bayesian-persuasion}) 
that depends on her realized action. 
We show that in this setting, optimal mixed (\Cref{thm:correlated-n-player-mixed-normal-form}) and pure (\Cref{thm:correlated-n-player-pure-normal-form}) commitments can both be computed efficiently with linear programming. (See \Cref{table:CE-Results} for a summary of these results.) 

In \Cref{sec:Bayesian}, we extend the previous settings to Bayesian games, i.e., games where players have private information about their own preferences. (See \Cref{table:Bayesian-Results} for a summary of the following results.) 
\begin{itemize}
    \item We show that in two-player games where only the follower has Bayesian types, the optimal commitment is \NPH{} to compute under standard complexity-theoretic assumptions, regardless of the availability of correlation devices (\Cref{cor:2-player-bayesian-follower-types-hardness}). 
    \item We show that in two-player games where only the leader has Bayesian types and correlation devices are not available, computing the optimal pure commitment is \NPH{} (\Cref{thm:2-player-pure-commitment-leader-types-hardness}), while the optimal mixed commitment can be computed efficiently with LP (\Cref{thm:two-player-mixed-strats-leader-types}).
    \item We show that in $n$-player games where only the leader has Bayesian types and correlation devices are available, optimal mixed commitments can be computed in polynomial time with LP (\Cref{thm:correlated-n-player-mixed-leader-types}).
\end{itemize}

We give only proof sketches in the main body of the paper. For the full proofs, see the appendices.

\subsection{Related Work}\label{subsec:related-work}

The work of \citet{rational-explanation-bribery} is most closely related to ours -- they essentially consider the two-player, mixed-commitment version of our setting. 
Their focus is on comparing the cases where the follower tiebreaks for and against the leader and characterizing the properties of optimal commitments in more detail. They also provide an alternative proof of our \Cref{Thm:2-player-mixed-strat}.

As discussed in the introduction, 
\citet{Stengel10:Leadership} study a two-player setting with mixed commitment over actions, but without commitment to payments. They characterize the range of possible payoffs for the players under optimal commitment and compare them to the payoffs from other ways of playing the game, such as playing a Nash or correlated equilibrium of the simultaneous-move game, or playing a commitment game with leadership roles reversed. 
\citet{conitzer_optimal_strategy_to_commit} were the first to study the computational complexity of computing optimal strategies in Stackelberg games. They consider both pure and mixed commitment in both Bayesian and normal-form games. 
Follow up work \cite{commitment-to-correlated-strats} introduced the idea of committing to \textit{correlated} strategies. 
Our present paper essentially extends the latter work to Bayesian games and extends both papers to allow the leader to commit to payments. 
Other works consider computing optimal commitments in extensive form \cite{Letchford2010} and Markov games \cite{Letchford2021}.
Further variations on the setting include mixed commitments to which multiple followers must respond with pure strategies \cite{coniglio2020computing-mixed-pure}, and games with multiple leaders who commit simultaneously \cite{castiglioni2019leader-multiple-leaders}.%

Another line of related work is that on $k$-implementation \cite{k-implementation, complexity-and-algs-k-implementation}, in which an outside \enquote{interested party} influences a normal-form game by committing to outcome-conditional payments, as in our setting. 
This is similar to our single commitment setting (\Cref{subsec:single-commitment}), but the rationality assumptions of $k$-implementation allow the followers to play any undominated strategy, while we make the stronger assumption that the followers play according to a Nash or correlated equilibrium.
\citet{internal_implementation} apply the $k$-implementation framework to games the interested party plays in herself. 
They essentially characterize commitments to payments that optimize the leader's pure security level in the induced game. 
In reinforcement learning, settings like policy teaching \cite[e.g.][]{policy-teaching} and reward poisoning \cite[e.g.][]{adaptive-reward-poisoning-RL} also study an outside party's ability to influence the agent's behavior via limited control over their reward function.

Payments have also been considered in the context of multi-agent reinforcement learning (MARL). 
\citet{Haupt2023} also consider outcome-conditional payments. However, their setting has one player propose payment contracts for everyone, leaving the remaining players merely with the decision of whether to accept or reject the proposed set of payments. In contrast, in our setting, the leader unilaterally commits to payments but only for herself. 
\citet{coco-Q} consider binding commitments to joint action profiles and side payments in two-player games, while other lines of work have considered (commitments to) sharing a fraction of one's reward \textit{without} conditioning on the outcome \cite{Willis2023,yi2022learning,kolle2023learning} and even entirely unconditional gifting of reward, i.e., the transfer of some constant amount of money \cite{Wang2021, lupu2020gifting}.

The principal--agent problem literature in economics, also called contract theory, \cite[e.g.,][]{Lambert1986, Demski1987, Stoughton1993, Core2002, Barron2003, Feess2004, Gromb2007, Malcomson2009} also considers settings in which one player (the principal) can pay another player (the agent) to incentivize that player to take an action which induces a distribution over \enquote{outcomes}, which the principal cares about. %
However, in this literature there is generally no notion of the principal taking actions herself, and hence no analog of commitment to actions. 
A key challenge is that the agent's action is generally unobservable except through its stochastic influence on the outcome, a difficulty our setting avoids. 
The principal-agent literature also generally assumes substantially more structure than the general normal-form games we consider.
There is typically a single agent who maximizes their payment minus the cost of their action (though they are not always risk neutral).
The rare settings with multiple agents typically have simple game structures, such as \cite[]{Barron2003, Gromb2007, OesConAcquisition} in the principal-expert sub-literature and \cite{combinatorial-agency, combinatorial-agency-optimal-contracts}, where the agents make binary decisions about whether to act and the outcome is either the AND or OR of these. 
Computational questions like those we consider are not typically the focus in principal-agent settings, though see \cite[][]{Babaioff-contract-complexity, simple-vs-optimal-contracts}, as well as \cite[][]{contracts-private-cost-per-unit-effort, contracts-under-moral-hazard-and-adverse-selection, Bayesian-agency-linear-vs-tractable-contracts} for more recent work focused on settings with Bayesian types (which we also consider).

Another economic setting, mechanism design, relies implicitly on the commitment ability of the designer and frequently assumes transferable utility, but assumes the designer has more expansive abilities to control the rules of interaction than we consider here.

One alternative to payments is outcome-conditional commitments to burn utility \cite{moulin1976_money_burning, money-burning-2x2}. For example, 
\citet{moulin1976_money_burning} characterizes two-player normal-form games in which the payoffs in completely mixed equilibria can be improved by commitments to burn utility. 
The benefits of burning utility here come from changing one's own incentives and hence the game's Nash equilibrium, something which doesn't apply to our setting because the leader commits to her actions.

\section{Preliminaries}
\subsection{Normal-Form Games}
An  $n$-player normal-form game $G$ is a pair $(A, u)$, where $A = A_1 \times A_2 \times \cdots \times A_n$ is a non-empty, finite set of pure action profiles (or outcomes) and $u: A \rightarrow \mathbb{R}^n$ is a utility function mapping each outcome to a vector containing utilities for each player.
Player $i$ has action space $A_i$ and utility function $u_i$. 
For $2$-player games, we will sometimes denote the action space $A\times B$. %

Let $\Delta(X)$ denote the set of probability distributions over a finite set $X$.
In a normal-form game, a strategy for Player $i$ is some $\sigma_i \in \Delta(A_i)$ and a strategy profile $\sigma$ is a vector of strategies for each player. 
A strategy is \textit{pure} if it plays one action with certainty. 
We'll use $-i$ to refer to the set of players besides Player $i$, so for example 
$\sigma_{-i}\in \bigtimes_{j \neq i} \Delta(A_{j})$
and $\sigma = (\sigma_i, \sigma_{-i})$. 
For convenience, we'll define utilities over strategy profiles $u$ in the obvious way: $u_i(\sigma) = \sum_{a \in A} \sigma(a) u_i(a)$, as we make the standard assumption that all players are risk-neutral expected utility maximizers.

An action $a_i$ is a \textbf{best response} to $\sigma_{-i}$ if $u_i(a_i, \sigma_{-i}) \geq u_i(a_i', \sigma_{-i})$ for all $a_i' \in A_i$.
A strategy $\sigma_i$ is a best response to $\sigma_{-i}$ if all supported actions are best responses. %
A \textbf{Nash equilibrium} (NE) %
of a game $G$ is a strategy profile $\sigma$ in which each $\sigma_i$ is a best response to $\sigma_{-i}$.
We use $\NE(G)$ to denote the set of Nash equilibria of a game $G$.
A \textbf{correlated equilibrium} (CE) \cite{aumann1987correlated} of a game $G$ is a distribution over outcomes $D \in \Delta(A)$ in which for each player $i$ and each $a_i$ with $D(a_i) > 0$, $a_i$ is a best response to $D(\cdot \mid a_i)\in \Delta(A_{-i})$.

\subsection{Commitment Games with Payments}
We study games in which a distinguished player (the \enquote{leader}) has the ability to modify a normal-form game by making a two-part commitment before it is played. 
First, she commits to taking a specific action or mixture over actions in the base game. 
Second, she can modify the payoff matrix by promising to transfer utility to another player whenever a certain outcome obtains. 
The leader's commitment is perfectly observed by all players. 

Formally, a commitment is a pair $(\sigma_1, P)$, where $\sigma_1$ is a strategy in the base game $G = (A, u)$ and $P: A \rightarrow \mathbb{R}^{n-1}_{\geq 0}$ is a payment function indexed $\{2, \dots, n\}$ that intuitively means the leader commits to transfer $P_i(a)$ utility to Player $i$ whenever outcome $a$ obtains.
A commitment induces a new normal-form game $G[\sigma_1, P]$ with $n-1$ players $\{2, \dots, n\}$, strategy space $A_{-1} = A_2 \times A_3 \times \dots \times A_n$, and utility functions $v_i^P(a_{-1}) = \sum_{a_1 \in A_1} \sigma_1(a_1) [u_i(a) + P_i(a)]$. 
We'll often drop the superscript $P$ from $v_i$ when it's clear from context, and we'll refer to players $i \geq 2$ as the followers.
While not a player in $G[\sigma_1, P]$, the leader still has utility function $v_1^P = u_1(a) - \sum_{i \geq 2}P_i(a)$ over the outcomes.
As with utility functions, we'll extend the domain of $P$ to include distributions over outcomes, i.e. $P(\sigma) = \sum_{a \in A} \sigma(a) P(a)$.  

We'll sometimes require $\sigma_1$ to be a pure action in $G$, which we refer to as the \textbf{pure commitment} case, and we'll sometimes allow $\sigma_1$ to be a mixture over actions in $G$, which we refer to as the \textbf{mixed commitment} case.%

When $n=2$, $G[\sigma_1, P]$ is a single-player game so the follower simply takes an action maximizing his expected utility. %
We assume he tiebreaks in favor of the leader, so the outcome of any game is clear. 
With more than two players, however, it's not clear what happens after the leader commits. We study two variants.

In the \textbf{single commitment} setting, the followers play $G[\sigma_1, P]$ simultaneously as a typical normal-form game.
This game may have multiple equilibria, so we'll assume the followers play the best Nash equilibrium for the leader. 
Hence, the leader's utility for a commitment $(\sigma_1, P_1)$ is well defined as
$$\max_{\sigma_{-1} \in NE(G[\sigma_1, P])} u_1(\sigma_1, \sigma_{-1}) - \sum_{j} P_{1,j}(\sigma_1, \sigma_{-1}).$$

In the \textbf{sequential commitment} case, Player $2$ becomes the leader of the (sequential) commitment game with base game $G[\sigma_1, P]$, which he plays optimally.\footnote{%
In other words, we consider the subgame perfect Nash equilibria of our sequential commitment setting, as is typical in Stackelberg games without payments \cite[e.g.,][]{vonStengel-mixed-commitment, conitzer_optimal_strategy_to_commit}.}
We assume that each player $i$ has a lexicographic tiebreaking rule: They first maximize their own utility, then that of player $i-1$, then $i-2$, etc.\footnote{We make this assumption to guarantee the existence of optimal commitments. For instance, if Player $2$ broke ties against Player $1$, Player $1$ would be incentivized to give Player $2$ a small additional payment $\eps$ so that Player $2$ would strictly prefer the action benefiting Player $1$, but no fixed $\eps$ could be optimal. In mechanism design, it is typically assumed that the agents break ties in favor of the leader for the same reason.} %
(We won't need to specify tiebreaking beyond this.)
Since the leader's utility for a commitment is well-defined in two-player games, in the $n\geq3$ sequential commitment case it's well-defined recursively.
The leader's optimal commitments are the $(\sigma_1, P_1)$ maximizing the leader's expected utility over the outcome of the sequential commitment game $G[\sigma_1, P_1]$.

\section{Nash Equilibria}\label{sec:NE}

We begin by considering the case where the leader does not have access to a correlation device and there are no private types. 
We start with the simplest setting, with only two players. 
Note that with mixed commitment, the leader can always achieve at least the utility of the best-for-her Nash equilibrium of the base game, simply by committing to her strategy in that Nash equilibrium and zero payments. 
In contrast, the leader's utility under her optimal pure commitment may or may not exceed her utility in a Nash equilibrium of the base game. 
For proofs of these claims%
, see 
\Cref{appendix:utiliity-implications}.

\subsection{Two Players}
We first consider the two-player case where the leader commits to payments and pure actions. 
In this setting, the optimal commitments can be computed efficiently using dynamic programming.

\begin{restatable}{theorem}{thmTwoPlayerPureCommitmentEasy}
\label{Thm:2-player-pure-commitment}    
    In a two-player game, the leader's optimal commitment to a payment function and a pure action can be computed in polynomial time. 
\end{restatable}
\begin{proofsketch}
    To incentivize the follower to play $a_2$ after committing to play $a_1$ herself, the leader would need to pay him exactly $\max_{a_2'} u_2(a_1, a_2') - u_2(a_1, a_2)$.
    Therefore, we can easily compute the leader's maximum utility for implementing any outcome $(a_1,a_2)$ and then maximize over them.
\end{proofsketch}

Optimal mixed commitments in two-player games can also be computed in polynomial time. 
This result has already been shown \cite[Corollary 11]{rational-explanation-bribery}, but we include it here for completeness.

\begin{restatable}{theorem}{thmTwoPlayerMixedCommitmentEasy}
\label{Thm:2-player-mixed-strat}    
    In a two-player game, the leader's optimal commitment to a payment function and a mixture over actions can be computed in polynomial time. 
\end{restatable}

\begin{proofsketch}
    For each follower action $a_2$, we construct an LP to compute the leader's optimal commitment that incentivizes $a_2$. 
    There are variables corresponding to the follower's expected payment when $a_2$ is played and to the leader's probabilities of playing each action.
    Constraints ensure $a_2$ gives the follower at least as much utility as any other action.
    We can then simply maximize over the $|A_2|$ LPs to find the leader's optimal commitment.
    This is an extension of \citet{conitzer_optimal_strategy_to_commit}'s approach from the version of the setting without payments. 
\end{proofsketch}

\subsection{More than Two Players, Single Commitment}\label{subsec:single-commitment}

We now consider games with $n\geq 3$ players, beginning with the single commitment case. 
Recall that in this setting, the leader makes a commitment $(\sigma_1, P_1)$ and the followers are assumed to play the best Nash equilibrium for the leader of the induced game $G[\sigma_1, P_1]$.

We give a strong negative complexity result: Even if the leader has only a single action, computing the optimal commitment for her to make is NP-hard, even in a $3$-player game. 
This immediately shows that both the pure and mixed commitment variants are NP-hard. 
Our proof is via reduction from the following problem:

\begin{definition*}[\textsc{Balanced Complete Bipartite Subgraph}]
    Given bipartite graph $G=(V,E)$ partitioned into partite sets $V_1$ and $V_2$ (s.t.\ $E$ consists only of edges between $V_1$ and $V_2$) and a natural number $k$, decide whether there exist subsets $V_1'\subseteq V_1$ and $V_2'\subseteq V_2$ s.t.\ $V_1'$ and $V_2'$ each have (at least) $k$ elements and there is an edge from each vertex in $V_1'$ to each vertex in $V_2'$.
\end{definition*}

\textsc{Balanced Complete Bipartite Subgraph} is NP-complete \cite[][page 446]{balanced-complete-bipartite-subgraph}.
It is sometimes called the \textsc{Balanced Biclique Problem}. 

\begin{restatable}{theorem}{thmThreePlayerSingleCommitmentHardness}
\label{thm:3-player-single-commitment-hardness}
    Consider an $n$-player game in which the leader commits to payments and a (mixture over) actions and then the remaining players play the best Nash equilibrium for the leader of the induced normal-form game. 
    Computing the leader's optimal commitment is \NP-hard, even for $n=3$ players and for games in which the leader has only a single action. 
\end{restatable}

\begin{proofsketch}

    We give a reduction from \textsc{Balanced Complete Bipartite Subgraph}.     
    The leader has only a single action, and roughly speaking, we design the game such that she cannot benefit from committing to payments. 
    Players 2 and 3 have strategy spaces corresponding to the vertices of the graph, and a strategy profile is an equilibrium giving the leader high utility if and only if it corresponds to balanced complete bipartite subgraph of size $k$. 
    Intuitively, this is because the followers benefit from playing adjacent vertices in their respective partite sets but cannot play any single vertex with probability more than $1/k$ without making themselves exploitable by the other follower.
    Our reduction is inspired by the reduction of \citet{gilboa1989nash-complexity} from \textsc{Clique} to \textsc{Best Nash} (deciding whether there exists a Nash equilibrium giving a certain player at least $k$ utility), but is simpler and reduces between different problems.
\end{proofsketch}

\subsection{More than Two Players, Sequential Commitment}

We now consider the case where the players commit sequentially. That is, first Player 1 commits to $(\sigma_1,P_1)$, then Player 2 commits to $(\sigma_2,P_2)$, and so on. As noted before, we assume that each player $i$ tiebreaks in favor of $i-1$, then $i-2$, etc. 
First, we give an efficient algorithm for the $n=3$ player pure commitment case, leaving the $n>3$ player case as an open problem. 

\begin{restatable}{theorem}{thmThreePlayerIteratedPureLP}
\label{thm:3-player-iterated-pure-LP}
    In a three-player game in which the players commit sequentially to payment functions and pure actions, the leader's optimal commitment can be computed in polynomial time. 
\end{restatable}

\begin{proofsketch}
    We construct an LP that computes, for any given outcome $a$, an optimal leader commitment that implements $a$. 
    To do so, we first show that, intuitively, it's optimal for the leader to pay Player $3$ some extremely large amount to minimize Player $2$'s utility if Player $2$ doesn't play $a_2$. 
    That is, for each action $a_2' \neq a_2$, the leader commits to make a large payment to Player $3$ when $(a_1, a_2', a_3^*)$ is played, where $a_3^* = \min_{a_3'} u_2(a_1, a_2', a_3')$.
    Since the LP has incentive constraints that ensure $a$ will be played, these large off-equilibrium payments will never actually need to be made. 
    The LP has $3$ variables corresponding to the on-equilibrium payments (including from Player $2$ to $3$) and incentive constraints ensuring that Player $2$ prefers to implement $a$ rather than either deviate himself or implement a deviation by Player $3$. 
    As usual, we then maximize over the $|A|$ possible outcomes. 
\end{proofsketch}

When players commit sequentially to \textit{mixtures} over actions, computing the optimal commitment is \NP-hard (as in the case without payments \cite{conitzer_optimal_strategy_to_commit}), even with only three players. 
We show this via reduction from the following problem.

\begin{problem}[\textsc{Balanced Vertex Cover}]\label{prob:balanced vertex cover}
    Given a graph $G = (V, E)$, decide whether there exists a subset of vertices $S \subseteq V$ of size at most $|V|/2$ such that, for all edges $e = (v_1, v_2) \in E$, at least one of $v_1$ and $v_2$ is in $S$.
\end{problem}
Given a subset of vertices $S$, we say an $e = (v_1, v_2)$ is \enquote{covered} if at least one of $v_1$ and $v_2$ is in $S$, and otherwise we say it is \enquote{uncovered}.
If $S$ covers all edges in $G$, it is called a \textit{cover} of $G$, or a $K$-cover if it has cardinality $K$. 

\textsc{Balanced Vertex Cover} is the special case of \textsc{Vertex Cover} in which the size of the requested cover is $|V|/2$. 
\textsc{Vertex Cover} was one of Karp's original 21 \NPC{} problems \cite{Karp-Original-NP-Completeness}. \citet{conitzer_optimal_strategy_to_commit} define
\textsc{Balanced Vertex Cover} and show it remains \NPC{}.
Intuitively, %
an arbitrary \textsc{Vertex Cover} instance can be reduced to a balanced instance by either adding isolated vertices while $K> |V|/2$ or by adding isolated triangles and increasing $K$ by $2$ per triangle while $K< |V|/2$.

\begin{restatable}{theorem}{thmIterativeMixedCommitmentHard}
\label{theorem:3-player-iterative-mixed-commitment-hard}
    In an $n$-player game in which the players commit sequentially to payment functions and mixtures over actions, computing the leader's optimal commitment is \NPH, even for $n=3$ players. 
\end{restatable}

\begin{proofsketch}
    We reduce from \Cref{prob:balanced vertex cover}, \textsc{Balanced Vertex Cover}. %
    We construct a game in which the first two players each have an action for each vertex and play \enquote{cooperatively} because they share the same utility function.
    They can achieve high utility if and only if the first mixes uniformly over vertices corresponding to a balanced vertex cover and the second mixes uniformly over its complement. 
    The third player has actions that \enquote{exploit} the first two if they don't play a balanced vertex cover, meaning the first player covers every edge and together they play every vertex with high enough probability.
    If no such exploit would succeed, the third player will instead play a different action giving the first two players high utility. %
    This approach adapts the reduction of \citet[Theorem 4]{conitzer_optimal_strategy_to_commit} to the version of the setting without payments, modifying their construction slightly so that the ability to commit to payments cannot benefit the first two players.    
\end{proofsketch}

\begin{table}[t]
\resizebox{\columnwidth}{!}{%
\begin{tabular}{|c|ll|}
\hline
 & \multicolumn{1}{l|}{\textbf{Pure Commitment}} & \textbf{Mixed Commitment} \\ \hline
$n=2$ &
  \multicolumn{1}{l|}{\begin{tabular}[c]{@{}c@{}} $\Theta(|A|)$ time via DP \\ (\Cref{Thm:2-player-pure-commitment})\end{tabular}} &
  \begin{tabular}[c]{@{}c@{}} 1 LP-solve per follower \\ action (\Cref{Thm:2-player-mixed-strat}) \end{tabular} \\ \hline
\begin{tabular}[c]{@{}l@{}}$n=3$,\\ single\end{tabular} &
  \multicolumn{2}{l|}{\begin{tabular}[c]{@{}c@{}} \NPH{}, even with only one leader action\\ (\Cref{thm:3-player-single-commitment-hardness}) \end{tabular}} \\ \hline
\begin{tabular}[c]{@{}c@{}}$n=3$,\\ sequential\end{tabular} &
  \multicolumn{1}{l|}{\begin{tabular}[c]{@{}c@{}} 1 LP-solve per action \\ profile (\Cref{thm:3-player-iterated-pure-LP})\end{tabular}} &
  \NPH{} (\Cref{theorem:3-player-iterative-mixed-commitment-hard}) \\ \hline
\end{tabular}%
}
\caption{Overview of Results from \Cref{sec:NE}}
\label{table:NE-results}
\end{table}

\section{Correlated Equilibria}\label{sec:CE}

We will now consider allowing the leader to commit to a signaling scheme (as in \cite{bayesian-persuasion})
 along with her payments and actions, inducing a correlated equilibrium. 
Specifically, the leader picks a probability distribution $D \in \Delta(A)$ and a payment function $P_j: A_j \rightarrow \mathbb{R}_{\geq 0}$ for each other player $j$. 
Intuitively, $P_j(a_j)$ is the amount paid to Player $j$ for following a recommendation to take action $a_j$. 
After committing to $D$, the leader privately draws an action profile $a \sim D$, sends each player $j\geq 2$ the private recommendation $a_j$, and plays $a_1$ herself. 
Finally, the followers simultaneously choose their actions in the base game, making no commitments of their own.
Note that the leader's commitment $(D, P)$ includes a commitment to actions since she commits to play according to the draw from $D$.

We give a revelation principle style result %
(\Cref{lemma:CE-signaling-revelation-principal-and-payments-simplification}) 
showing that any leader commitment
is equivalent (in terms of the distribution over outcomes and payoffs) to an \textbf{incentive compatible} commitment in which all followers are incentivized to follow their recommended actions. 
Hence, we will assume without loss of generality that the leader's commitment is incentive compatible.
We also show that allowing the leader to send arbitrary messages makes no difference compared to allowing only action recommendations and that allowing payments to depend on the full outcome and recommendation profile is no more powerful than simply paying players for following their recommended action.  
(Again, see \Cref{lemma:CE-signaling-revelation-principal-and-payments-simplification} for details.)

Formally, a commitment $(D, P)$ is \textbf{incentive compatible} if, for any player $j$ and any action $a_j$ with $D(a_j)>0$, player $j$ is always (weakly) incentivized to follow a recommendation to play $a_j$, assuming all other players follow their recommendations. 
That is, for all players $j$, actions $a_j$ with $D(a_j) > 0$, and all $a_j' \neq a_j$, $$\sum_{a_{-j}} D(a_{-j}| a_j) u_j(a_j, a_{-j}) + P_j(a_j) \geq \sum_{a_{-j}} D(a_{-j}| a_j) u_j(a_j', a_{-j}).$$
We call these the incentive constraints for player $j$. 
Hence, the problem of computing the leader's optimal commitment is equivalent to that of finding a utility maximizing distribution $D$ over outcomes and payment function $P$, subject to the incentive constraints. 
(Note that there is no incentive constraint for the leader since she commits to her action, though she is sometimes constrained to play a pure action.)

It turns out that an optimal commitment to a signaling scheme can be found in polynomial time regardless of whether the leader can commit to mixtures or only to pure actions, as the following the results show.

\begin{restatable}{theorem}{thmCorrelatedMixedActionNormalForm}
\label{thm:correlated-n-player-mixed-normal-form}
    In an $n$-player game, the leader's optimal commitment to a payment function, mixture over actions, and signaling scheme can be computed in polynomial time.
\end{restatable}
\begin{proofsketch}
    We construct a linear program which computes the leader's optimal incentive-compatible commitment. Such commitments are optimal by Lemma C.1. %
    Our LP is similar to the feasibility LP for a correlated equilibrium except that it incorporates payments.
    In addition, the leader's incentive constraints are dropped because she commits to follow her action recommendation and the objective is the leader's expected utility (after payments).

    The LP has variables $p_a$ for each action profile $a$ representing the that action profile is recommended and variables $t_i(a_i)$ representing the payment from the leader to each follower $i$ when an recommendation to play action $a_i$ is followed. 
    Constraints ensure that, for each follower and each possible action recommendation, the follower's expected utility for following the recommendation (including payments) is at least their expected utility for any other action.
    This expectation is taken over the player's uncertainty over the other players' action recommendations given their own. 
    (The constraints are designed such that those corresponding to actions that are never recommended are trivially satisfied.)

    Our approach is similar to and inspired by that of \citeauthor{commitment-to-correlated-strats}
    \cite{commitment-to-correlated-strats} for computing optimal correlated distributions to commit to without payments.
    \end{proofsketch}

\begin{restatable}{theorem}{thmCorrelatedPureActionNormalForm}\label{thm:correlated-n-player-pure-normal-form}
    In an $n$-player game, the leader's optimal commitment to a payment function, pure action, and signaling scheme can be computed in polynomial time.%
\end{restatable}
\begin{proofsketch}
    For any leader action $a_1$, we can add a constraint to the LP from \Cref{thm:correlated-n-player-mixed-normal-form} to require the leader to play the pure strategy $a_1$, which allows us to compute the leader's optimal commitment and corresponding utility when committing to $a_1$. 
    Solving such an LP for all actions $a_1$ and maximizing over their values gives the leader's optimal commitment overall. 
\end{proofsketch}

\begin{table}[t]
\resizebox{\columnwidth}{!}{%
\begin{tabular}{|c|c|c|}
\hline
        & \textbf{Pure Commitment}                             & \textbf{Mixed Commitment}                                 \\ \hline
Any $n$ & \begin{tabular}[c]{@{}c@{}}
1 LP-solve per leader action \\ (\Cref{thm:correlated-n-player-pure-normal-form})\end{tabular} & \begin{tabular}[c]{@{}c@{}}1 LP-solve \\ (\Cref{thm:correlated-n-player-mixed-normal-form})\end{tabular} \\ \hline
\end{tabular}%
}
\caption{Overview of Results from \Cref{sec:CE}}
\label{table:CE-Results}
\end{table}

\section{Bayesian Games}\label{sec:Bayesian}

We now consider settings with Bayesian, rather than normal-form, base games. 
In normal-form games, the agents' utility functions are common knowledge. 
Bayesian games, in contrast, model agents with private information about their own preferences.
In Bayesian games, each player $i$ has a set of types $\Theta_i$ and for each type $\theta_i$, a utility function $u_i^{\theta_i}: A \rightarrow \mathbb{R}$.
In games in which only one player has multiple types, we often drop the subscript $i$ from $\theta$ and $\Theta$.

The distribution $\pi_i \in \Delta(\Theta_i)$ over each player's type is common knowledge before the game begins (\textit{ex ante}). 
Then, \textit{ex interim}, each player $i$ learns their own realized type $\theta_i$ and chooses an action.
Hence, a strategy for Player $i$ in a Bayesian game is a mapping $\sigma_i : \Theta_i \rightarrow \Delta(A_i)$ which specifies a mixture over actions for each potential type. 
An action $a_i$ is a best response for type $\theta_i$ if $\E_{\theta_{-i} | \theta_i }[u_i^{\theta_i}(a_i, \sigma_{-i}(\theta_{-i}))] \geq \E_{\theta_{-i} | \theta_i }[u_i^{\theta_i}(a_i', \sigma_{-i}(\theta_{-i}))]$ for all $a_i' \in A_i$, and a strategy $\sigma_i$ is a best response to $\sigma_{-i}$ if all actions supported in each $\sigma_i(\theta_i)$ are best responses for $\theta_i$. 
A strategy profile $\sigma$ is a \textbf{Bayes-Nash equilibrium} (BNE) if each $\sigma_i$ is a best response to $\sigma_{-i}$.

It turns out that leader and follower types introduce different considerations and have different impacts on the complexity of computing optimal commitments, so we'll consider the two cases separately. 
(When we refer to a setting with leader types, we mean the special case of Bayesian games in which the follower has only one type, and vice versa.) %
In addition, we have all the same variations in the Bayesian setting as we did without private information: pure vs mixed commitment, sequential vs single commitment, and whether or not the leader is able to commit to a signaling scheme. 

It's important to note that the commitments are to potentially different (mixtures over) actions for each type. 
In particular, in the pure commitment case, each type must play a pure action, but different types can play different actions. 
If the leader's actions couldn't depend on her type, this setting would be equivalent to one without leader types in which the leader's utility was equal to her expected utility over types.
We also make the natural assumptions that the leader's payment cannot depend on the follower's type and that the followers don't gain any information about the leader's realized type before taking their action (except insofar as their action recommendations are correlated with the leader's type). 

\subsection{Follower Types}\label{subsec:follower-types}

First, consider the case with $n=2$ agents where the leader has only one type but has Bayesian uncertainty over the follower's type. 
Without payments, computing the optimal pure action for the leader to commit to is quite simple: for each possible leader action, one can compute each follower type's best response and hence the leader's expected utility \cite[Theorem 6]{conitzer_optimal_strategy_to_commit}.

However, the problem becomes difficult if the leader can additionally commit to payments. 
We show a surprising connection to auction theory, reducing to our present setting from problem of finding a revenue-maximizing item pricing for $m$ items and a single unit-demand buyer.

A unit-demand buyer is one who essentially \enquote{wants} at most one item, i.e. considers the items perfect substitutes.  
Formally, a unit-demand buyer has a value $v_i$ for each item $i$ and their utility for receiving a set of items $S$ is the maximum over their values for items in the set $\max_{i\in S} v_i$. 
An item pricing (or posted pricing) offers each item to the buyer at a take-it-or-leave-it price and allows the buyer to purchase whatever items they want, as at a typical retail store. 
Formally, this is a vector $r$ of non-negative prices, one for each item. 
Therefore, given an item pricing, a unit-demand buyer will purchase a single item maximizing his value minus the price or, if all prices are greater than the items' values to the buyer, will purchase nothing. 

\begin{restatable}[\textsc{Unit Demand Item Pricing}]{problem}{probUnitDemandItemPricing}\label{prob:unit-demand-item-pricing}
    Consider a finite-support distribution $D$ of value vectors $v \in \mathbb{R}^m_{\geq 0}$. %
    For an item pricing vector $r \in \mathbb{R}^m_{\geq 0}$, let $i^*(v) \in \arg\max_i [v_i - r_i]$ be the most favorable item for a buyer with values $v$ to buy, tiebreaking in favor of the item with highest value $v_i$.
    Let $A(r) = \{ v \in \supp(D) |  v_{i^*(v)} - r_{i^*(v)} \geq 0 \}$ be the set of values $v$ for which the buyer buys an item. 
    Decide whether there exists $r$ such that the seller's revenue
    $\Rev(r) = \sum_{v \in A(r)}[ D(v) \cdot r_{i^*(v)}]$ is at least $K$.
\end{restatable}

\citet{auctions-hardness} gives a hardness result for a special case of \textsc{Unit Demand Pricing} they define, which we'll call \textsc{Unit Demand Pricing for Uniform Budgets}.
The special case requires each value vector $v$ to have some \enquote{budget} $\beta_v$ such that all entries $v_i \in \{0, \beta_v\}$. 
They refer to it as the unit-demand min-buying (or envy-free) pricing problem with uniform budgets, economist's version. %
Specifically, they show \textsc{Unit Demand Pricing for Uniform Budgets} is \NPH{} to approximate subpolynomially in the number of items unless \NP{} is in bounded-error probabilistic sub-exponential time:

\begin{theorem*}[\cite{auctions-hardness}, Theorem 5]
    \textsc{Unit Demand Pricing for Uniform Budgets} is hard to approximate within $O(|G|^\eps)$ for some $\eps > 0$ if \NP{}
    $\not\subseteq \cap_{\delta > 0} \BPTIME(2^{O(n^\delta)})$.
\end{theorem*}
The general problem of \textsc{Unit Demand Item Pricing} is widely believed to unconditionally be \NPH{} \cite{shuchi-algorithmic-mechanism-design-survey}, perhaps in part because the version where the buyer's values for items are independent rather that correlated is \NPH{} \cite{unit-demand-pricing-independent-hardness}.

We now give a reduction from general \textsc{Unit Demand Item Pricing} to our present setting of computing optimal pure action and payment commitments in $2$-player Bayesian games. We'll therefore inherit the hardness result from the uniform budgets special case.  

\begin{restatable}{theorem}{thmTwoPlayerFollowerTypesHardness}
\label{thm:2-player-follower-types-hardness}
    \textsc{Unit Demand Item Pricing} is polynomial-time reducible to the problem of computing the leader's optimal payments and (mixtures over) actions in a two-player Bayesian game with follower types only. 
    Further, it is reducible to instances in which the leader has only a single action. 
\end{restatable}

\begin{proofsketch}
    We reduce an arbitrary instance $U$ of \textsc{Unit Demand Item Pricing} to a $2$-player Bayesian game $G$ with a single leader type and single leader action. 
    In $G$, the follower has actions $b_i$ corresponding to each item in $U$ and types $\theta_v$ corresponding to each value vector $v$ in $U$. Each type $\theta_v$ occurs in $G$ with the same probability as $v$ in $U$. 
    We construct payoffs such that taking action $b_i$ corresponds to purchasing item $i$ in $u$ and the leader's payment function corresponds to an item pricing in $U$. 
    Specifically, the leader receives some very large utility $Z$ whenever the follower plays any action $b_i$, but the follower of type $\theta_v$ receives utility $-Z + v_i$ for playing $b_i$.
    Committing to a payment of $Z-r_i$ for action $b_i$ then corresponds to setting a posted price of $r_i$: if $b_i$ is played the leader gets utility $Z - (Z -r_i) = r_i$ and the follower gets utility $(-Z+ v_i)+(Z - r_i) = v_i - r_i$. 
    Because the follower in $G$ plays the single action that maximizes his utility, his behavior is equivalent to a unit demand buyer in $U$, who purchases the single item which maximizes his utility. 
\end{proofsketch}

\begin{corollary}\label{cor:2-player-bayesian-follower-types-hardness}
    In a two-player Bayesian game with follower types only, computing the leader's optimal commitment to a payment function and a (mixture over) actions is \NPH, even if the leader has only a single action, assuming 
    \NP{} $\not\subseteq \cap_{\delta > 0} \BPTIME(2^{O(n^\delta)})$. 
\end{corollary}
Because our construction requires only a single leader action, it immediately shows the hardness of both the pure and mixed commitment versions of the problem. 
Likewise, with only $n=2$ players a single leader action, there is clearly no possible correlation to be had, and so access to signaling devices makes no difference. 
Since this hardness result applies to any version of the setting with follower types, we now consider settings with leader types only.

\subsection{Leader Types}\label{subsec:leader-types-NE}

When the leader has types, it's not immediately obvious whether to model the leader as already knowing her type at the time of commitment: Should the leader make her commitment \textit{ex ante} or \textit{ex interim}?
However, if the leader with types commits \textit{ex interim} with a realized type $\theta_i$, the problem instance is equivalent to one \textit{without} leader types in which the leader's utility function is $u_1^{\theta_i}$:
Once the leader commits to her strategy, the follower(s) simply best respond and the leader's utility function (and hence her type and uncertainty over it) has no further impact on the game.\footnote{Except that we assume the follower(s) break ties in favor of the leader's realized type.} 

Therefore, we'll assume the leader seeks to maximize her \textit{ex ante} expected utility, where the expectation is over the randomness in her realized type (as well as the randomness in the game, of followers' types, etc.). 
In this setting, a commitment is a payment function and a mapping from types to actions.
Note that it would make no difference if the leader could commit to different payment functions for different types: The follower's followers' actions and leader's \textit{ex ante} expected utility depend only on the expected payment function (over leader types).

Unlike with follower types, the pure and mixed commitment settings with leader types differ in hardness. We'll begin with the former. 
Even for $n=2$ players, the analog of the problem without payments is \NP-hard \cite[Theorem 5]{conitzer_optimal_strategy_to_commit}. 
We now show this hardness continues to hold when the leader can commit to payments.

\begin{restatable}{theorem}{thmLeaderTypesPureCommitmentTwoPlayerHardness}
\label{thm:2-player-pure-commitment-leader-types-hardness}
    In a two-player Bayesian game with leader types only, computing the leader's \textit{ex ante} optimal commitment to a payment function and an action for each type is \NPH, regardless of whether she has the ability to commit to a signaling scheme.  
\end{restatable}

\begin{proofsketch}
    We prove this via reduction from \textsc{Vertex Cover}.
    The leader has $K$ types, each of which occur with equal probability, and an action $a_v$ corresponding to each vertex $v$ in the graph. 
    We show that she achieves strictly positive utility if and only if she commits to a strategy where, for each vertex $v$ in a $K$-cover of $G$, $a_v$ is the action of one of her types.
    The follower has strategies $b_e$ corresponding to each edge in $e \in E$ and an additional strategy $b_0$.
    All leader types get utility $1$ if the follower plays $b_0$ and $0$ otherwise. 

    The follower's utility function is such that, if all edges are covered by some leader type, the follower will prefer to play $b_0$, resulting in a leader utility of $1$. 
    However, if there exists an \enquote{uncovered} edge $e$ for which no leader type plays a vertex $v \in e$, the follower will prefer to play $b_e$ rather than $b_0$, resulting in a leader utility of $0$. 
    The leader cannot pay the follower enough to incentivize him to play $b_0$ when an uncovered edge exists while still getting positive utility herself, and signaling cannot help incentivize the follower to play $b_0$.
    Therefore, the leader can achieve strictly positive utility if and only if there exists a $K$-cover in the \textsc{Vertex Cover} instance. 

    Our reduction is the same as that of \citet{conitzer_optimal_strategy_to_commit} to the version of this setting without payments or signaling.
\end{proofsketch}

In contrast to the pure commitment case, the mixed commitment case with leader types is tractable with $n=2$ players. 

\begin{restatable}{theorem}{thmTwoPlayerLeaderTypesMixedStrategiesPolytime}
\label{thm:two-player-mixed-strats-leader-types}
    In a two-player Bayesian game with leader types only, the leader's \textit{ex-ante} optimal commitment to a payment function and a mixture over actions for each type can be computed in polynomial time. 
\end{restatable}

\begin{proofsketch}
    We extend our LP approach from \Cref{Thm:2-player-mixed-strat}, which considers the analogous setting without types. Our LP variables now encode a distribution over actions for each leader type and optimizes the leader's \textit{ex ante} expected utility over her types, again subject to the constraint that the follower is incentivized to play a particular pure action. 
    As usual, we then maximize over the $|B|$ follower actions. 
\end{proofsketch}

If the leader doesn't have access to a signaling device, the problem becomes hard for $n\geq 3$: 
Even without leader types, we've already given hardness results in \Cref{thm:3-player-single-commitment-hardness} and \Cref{theorem:3-player-iterative-mixed-commitment-hard} for both single and sequential mixed commitment, respectively. 
Hence, we turn to the case where the leader can commit to a signaling scheme.
With signaling, the problem becomes tractable for any number of players $n$. 

\begin{restatable}{theorem}{thmCELeaderTypesMixedAction}
\label{thm:correlated-n-player-mixed-leader-types}
    In an $n$-player Bayesian game with leader types only, the leader's \textit{ex ante} optimal commitment to a payment function, a mixture over actions for each type, and a signaling scheme can be computed in polynomial time.
\end{restatable}
\begin{proofsketch}
    We extend our linear programming approach from \Cref{thm:correlated-n-player-mixed-normal-form}, which considers the analogous setting without types. Our LP variables now encode a different distribution over outcomes for each leader type, as well as payments. 
    We optimize the leader's \textit{ex ante} expected utility over her types, and the incentive constraints now hold in expectation over the leader's types. 
\end{proofsketch}

One may recall that for the normal-form case with access to signaling devices, we could extend the mixed commitment approach to the pure setting by simply having one LP for each leader action. 
Taking the same approach doesn't work with leader types, however, because commitments to pure actions are now functions from leader types to actions, rather than just single actions.
Hence, we'd need one LP for each of the $\Omega(|A_1|^{|\Theta|})$ commitment functions rather than just one per leader action. 
Indeed, with pure commitment and leader types, we already have the hardness result from the $n=2$ case (\Cref{thm:2-player-pure-commitment-leader-types-hardness}). %

\begin{table}[]
\resizebox{\columnwidth}{!}{%
\begin{tabular}{|cc|rc|}
\hline
\multicolumn{2}{|c|}{} &
  \multicolumn{1}{c|}{\textbf{Pure Commitment}} &
  \textbf{Mixed Commitment} \\ \hline
\multicolumn{1}{|c|}{\begin{tabular}[c]{@{}c@{}} \textbf{Follower} \\ \textbf{Types}\end{tabular}} &
  $n=2$ &
  \multicolumn{2}{c|}{\begin{tabular}[c]{@{}c@{}}\NPH{} under standard complexity-\\ theoretic assumptions (\Cref{cor:2-player-bayesian-follower-types-hardness})\end{tabular}} \\ \hline
\multicolumn{1}{|c|}{\multirow{3}{*}{\begin{tabular}[c]{@{}c@{}} \textbf{Leader} \\ \textbf{Types}\end{tabular}}} &
  $n=2$ &
  \multicolumn{1}{r|}{\multirow{3}{*}{  \NPH{} (\Cref{thm:2-player-pure-commitment-leader-types-hardness})
  }} &
  \begin{tabular}[c]{@{}c@{}}1 LP solve per follower\\ action (\Cref{thm:two-player-mixed-strats-leader-types})\end{tabular} \\ \cline{2-2} \cline{4-4} 
\multicolumn{1}{|c|}{} &
  \begin{tabular}[c]{@{}c@{}}$n=3$, \\ no signaling\end{tabular} &
  \multicolumn{1}{r|}{} &
  \begin{tabular}[c]{@{}c@{}}\NPH{} even without types\\ (\Cref{thm:3-player-single-commitment-hardness}, \Cref{theorem:3-player-iterative-mixed-commitment-hard})\end{tabular} \\ \cline{2-2} \cline{4-4} 
\multicolumn{1}{|c|}{} &
  \begin{tabular}[c]{@{}c@{}}Any $n$, \\ signaling\end{tabular} &
  \multicolumn{1}{r|}{} &
  1 LP solve (\Cref{thm:correlated-n-player-mixed-leader-types}) \\ \hline
\end{tabular}%
}
\caption{Summary of Results from \Cref{sec:Bayesian}}
\label{table:Bayesian-Results}
\end{table}

\section{Conclusion}

In this work, we've presented a model combining Stackelberg games with outcome-conditional utility transfers and analyzed the computational complexity of computing optimal commitments. 
We've varied our setting along several dimensions: whether the leader commits to mixed or pure actions, whether the leader can commit to a signaling scheme, whether the followers play simultaneously or make commitments sequentially, and whether players have private information about their payoffs.  
We've given a mixture of efficient algorithms for computing optimal commitments, primarily via linear programming, and \NP-hardness results. 
For an overview of the specific results, we refer back to \Cref{table:NE-results,,table:CE-Results,,table:Bayesian-Results}.

There are many open directions for future work. 
One natural idea is to extend the framework of commitments to both actions and payments to game representations besides normal form. 
Another is to vary the assumptions we make about the payments. 
For instance, one could consider settings where there are costs associated with commitment, akin to the legal costs of writing a contract. 
One could also consider restrictions on what the payments can depend on, for instance if %
only certain actions can be detected and thus have payments associated with them.  

One could also consider weakening the strong, completely binding model of commitment we've studied in this work.  
For instance, agents might only be able to commit to certain aspects of actions, or they might be able to commit \textit{against} taking certain pure actions without being able to commit to a mixture over the remaining actions (as in \cite{disarmament-games}).
Examples of weaker levels of commitment are prevalent in economic and societal interactions.  
For instance, firms give press releases and may fear the public relations or stock price ramifications of reversing course. 
Similarly, they make investments that make the indicated actions very likely, but not entirely certain.

\begin{acks}
We thank Vincent Conitzer for helpful discussions. 
Nathaniel Sauerberg acknowledges financial support from the Center on Long-Term Risk and ML Alignment and Theory Scholars. Caspar Oesterheld acknowledges financial support through grants from the Cooperative AI Foundation, Polaris Ventures (formerly the Center for Emerging Risk Research), and Jaan Tallinn's donor-advised fund at Founders Pledge. Additionally he is grateful for support by an FLI PhD Fellowship.
\end{acks}

\bibliographystyle{ACM-Reference-Format} 
\bibliography{refs2}

\clearpage
\appendix

\section{Deferred Proofs of Utility Claims}\label{appendix:utiliity-implications}
\begin{table}[b]
\begin{tabular}{cccc}
                            & Left                      & Middle                          & Right                     \\ \cline{2-4} 
\multicolumn{1}{c|}{Top}    & \multicolumn{1}{c|}{$-1$, $0$} & \multicolumn{1}{c|}{$-1$, $-2$} & \multicolumn{1}{c|}{$-1$, $0$} \\ \cline{2-4} 
\multicolumn{1}{c|}{Bottom} & \multicolumn{1}{c|}{$\hphantom{-}2$, $0$} & \multicolumn{1}{c|}{$0$, $2$}       & \multicolumn{1}{c|}{$\hphantom{-}0$, $1$}  \\ \cline{2-4} 
\end{tabular}%
\caption*{Table 1: A game in which commitment to strategies and payments is beneficial to the leader.}
\label{ex:intro-appendix-copy}
\end{table}

\begin{proposition}
\label{lemma:analysis-intro-ex}
    In the game in Table 1, the leader can achieve utility $\nicefrac{1}{3}$ via commitment to a mixture over actions and payments. 
    However, she cannot achieve strictly positive expected utility either by committing to actions only or in any Nash equilibrium of the game induced by a commitment to payments only. 
\end{proposition}

\begin{proof}
    First, observe (as was shown in the introduction) that the leader can achieve a utility of $\nicefrac{1}{3}$ by committing to play the mixture $\left(\nicefrac{1}{3}, \nicefrac{2}{3}\right)$ over (Top, Bottom) and to the payment function $P(a) = \begin{cases}
        1 & \text{ if } a = (B, L) \\
        0 & \text{ otherwise}
    \end{cases}$.
    Then, the follower's expected utility for playing Right is %
    $\nicefrac{2}{3}$%
    , as is his utility for playing Left since the their payoffs after payments are the same.
    His utility for playing Middle is $2\cdot{}\nicefrac{2}{3} -2\cdot{}\nicefrac{1}{3} = \nicefrac{2}{3}$. 
    Hence, the follower is indifferent between each of his actions and tie-breaks in favor of the leader, playing Left for a leader payoff of $-1\cdot{}\nicefrac{1}{3}+(2-1)\cdot{}\nicefrac{2}{3} = \nicefrac{1}{3}$.

    Next, we claim that with payments alone, the leader cannot achieve a strictly positive expected utility in any Nash equilibrium of a game $G[P]$ induced by a commitment to payments.
    ($G[P]$ is the normal-form game with the same action space as $G$ and with utilities $v_1 = u_1 - P$, $v_2 = u_2-P$, where $-$ is pointwise and $u$ is the utilities in $G$. 
    
    First observe that the leader can only achieve strictly positive utility if the outcome (Bottom, Left) is played. 
    Therefore, she must incentivize the follower to play Left with positive probability. 
    Notice that if the leader always played Bottom, then the follower would prefer Middle to Left unless $P(B, L) \geq 2$, which would leave the leader unable to achieve positive utility. 
    Therefore, the leader must play Top with nonzero probability, and hence Top must be a best response to the follower's strategy in $G[P]$. 
    But since playing Top always results in utility $-1$ for the leader, the leader cannot achieve positive utility at all.     

    Finally, we claim that with commitment to actions alone, the leader cannot achieve a strictly positive expected utility. 
    Again, the leader can only achieve strictly positive utility if the follower plays Left.
    The follower strictly prefers playing Right to Left unless Top is played with probability $1$, so the leader must do so. 
    However, if she always plays Top, the leader cannot get positive utility regardless of P2's action.
\end{proof}

Next, we prove in more detail the two claims made in \Cref{sec:NE} about when commitment benefits the leader. %
First, we show that making an optimal commitment to payments and mixtures over actions always benefits the leader relative to playing any Nash equilibrium of the base game. 
In contrast, optimal commitment to payments and a pure action benefits the leader in some games and hurts her in others. 

\begin{proposition}
    In any game $G$, the leader's expected utility under her optimal commitment to a payment function and mixture over actions is weakly higher than her expected utility under any Nash equilibrium of the base game, and there exist games where this inequality is strict.
\end{proposition}

    Note that regarding the first part of the claim, \cite{vonStengel-mixed-commitment} and later \cite{commitment-to-correlated-strats} have shown the stronger result that the leader's utility under optimal commitment to a mixture over actions is at least her utility in any \textit{correlated } equilibrium of the base game.     
    The second part of the claim has already been demonstrated for commitment to mixtures without payments in \cite{conitzer_optimal_strategy_to_commit}.

\begin{proof}
    For the first part of the claim, consider any game $G$ and
    suppose the leader can achieve a utility of $u^*$ in some Nash equilibrium $\sigma$ of $G$.  
    Then we claim that the leader can achieve a utility of at least $u^*$ by committing to the mixture $\sigma_1$ and the zero payment function $P$.
    Note that any action supported in $\sigma_{2}$ is optimal for the follower in $G[\sigma_1,P]$ because $\sigma$ is a Nash equilibrium in $G$, ie $\supp(\sigma_2) \subseteq \BR(\sigma_1)$. (Where $\BR(\cdot)$ denotes the set of best responses.)
    Since the follower tiebreaks among his best responses in favor of the leader, the leader's utility is $$\max_{a_2 \in \BR(\sigma_1)} u_1(\sigma_1, a_2) \geq \max_{a_2 \in \supp(\sigma_2)} u_1(\sigma_1, a_2) \geq u_1(\sigma_1, \sigma_2) = u^*,$$ as desired.

    The game in Table 1 demonstrates the second part of the claim.
    As the analysis in \Cref{lemma:analysis-intro-ex} shows, the leader can achieve a utility of $\nicefrac{1}{3}$ by a commitment to a mixture over actions and a payment function. 
    However, the only Nash equilibrium of the base game is (Bottom, Middle), which would give the leader utility $0$. 
    This can be verified by observing that Bottom is a strictly dominant strategy for Player $1$ and Middle is Player $2$'s unique best response. 
\end{proof}

In contrast, if she doesn't know the game to be played, a leader can't say ahead of time whether she'd prefer to just play the base game simultaneously or whether she'd prefer to commit ahead of time to a payment function and pure action.

\begin{table}[b]
\begin{tabular}{lll}
                            & Left                          & Right                         \\ \cline{2-3} 
\multicolumn{1}{l|}{Top}    & \multicolumn{1}{l|}{$3$, $0$} & \multicolumn{1}{l|}{$1$, $1$}  \\ \cline{2-3} 
\multicolumn{1}{l|}{Bottom} & \multicolumn{1}{l|}{$2$, $1$} & \multicolumn{1}{l|}{$0$, $0$} \\ \cline{2-3} 
\end{tabular}%
\caption{A game in which commitment to payments and pure actions benefits the leader.}
\label{ex:commit-to-dominated-strat}
\end{table}

\begin{table}[b]
\begin{tabular}{lll}
                           & Heads                          & Tails                          \\ \cline{2-3} 
\multicolumn{1}{l|}{Heads} & \multicolumn{1}{l|}{$1$, $-1$} & \multicolumn{1}{l|}{$-1$, $1$}  \\ \cline{2-3} 
\multicolumn{1}{l|}{Tails} & \multicolumn{1}{l|}{$-1$, $1$}  & \multicolumn{1}{l|}{$1$, $-1$} \\ \cline{2-3} 
\end{tabular}%
\caption{Matching Pennies, a game in which commitment to payments and pure actions hurts the leader.}
\label{ex:matching-pennies}
\end{table}

\begin{proposition}
    There exist games for which the leader's expected utility under her optimal commitment to a payment function and a pure action is strictly higher than her expected utility in any Nash equilibrium of $G$. 
    There also exist games for which the leader's expected utility under her optimal commit to a payment function and a pure action is strictly lower than in the worst (for her) Nash equilibrium of the base game. 
\end{proposition}

\begin{proof}
    To see the first part of the claim, consider the game $G$ in \Cref{ex:commit-to-dominated-strat}.
    Note that in $G$, Top is a strictly dominant strategy for the leader, and hence the only Nash equilibrium of the base game is (Top, Right), which results in a leader utility of $1$.
    However, if the leader commits to playing Bottom and making no payments, the follower's best response is Left, which results in a leader utility of $2$. 

    The second part of the claim can be seen by considering a zero sum game $G$ such as matching pennies (\Cref{ex:matching-pennies}).
    In this game, the only Nash equilibrium is for both players to mix uniformly, achieving an expected payoff of $0$. 
    If forced to commit a pure action, however, the leader can only achieve a utility of $-1$ regardless of her payments. 
    She must commit to playing either Heads or Tails (they are symmetric), which allows the follower to achieve a utility of at least $1$ by playing the opposite action, regardless of the payments. 
    Since the game (including payments) is zero sum, this means the leader cannot get a utility more than $-1$. 
\end{proof}

\section{Deferred Proofs from Section~\ref{sec:NE}}\label{appendix:NE-proofs}

\subsection{Two Players}

\thmTwoPlayerPureCommitmentEasy*
\begin{proof}
    Consider an outcome $(a_1, a_2)$. 
    To implement this, the leader must commit to $a_1$ and to payments such that $v_2(a_1 a_2) = u_2(a_1, a_2) + P(a_1, a_2) \geq u_2(a_1, a_2') + P(a_1, a_2')$ for any $a_2'$. 
    Hence, the leader must pay at least $\max_{a_2'} u_2(a_1, a_2') - u_2(a_1, a_2)$ to incentivize the follower to play $a_2$. The payment function $P$ which pays exactly that amount when $(a_1 a_2)$ obtains and $0$ otherwise is sufficient to implement $a$. 
    Therefore, the leader's optimal utility when implementing $(a_1, a_2)$ is $v_1^*(a_1, a_2) \coloneqq u_1(a_1, a_2) - [\max_{a_2'} u_2(a_1, a_2') - u_2(a_1, a_2)]$.

    Therefore, to compute the leader's maximum achievable utility and a commitment achieving it by, for each $a_1$, we can compute $\max_{a_2'} u_2(a_1 a_2')$ and then $v_1^*(a_1, a_2)$ for each $a_2$. 
    Simply maximizing over all outcomes gives the solution in linear time. 
    Note also that any correct algorithm must look at each entry of the payoff matrix at least once and hence run in $\Omega(|A|)$, so this algorithm is optimal.
\end{proof}

\thmTwoPlayerMixedCommitmentEasy*
\begin{proof}
    Note that we can assume without loss of generality that the follower plays a pure action: 
    Since the leader has already committed, any mixture over best responses (after tiebreaking in favor of the leader's utility) is optimal, and assuming the follower chooses a pure action makes no difference to either player's utility. 
    
    Our approach is to compute, for a given follower action $b$, the optimal leader commitment for which $b$ is a best response. 
    This gives us the maximum utility the leader can achieve while implementing $b$.
    We then simply maximize over the $|B|$ actions for the follower to find an optimal leader's commitment and its resulting utility.

    We claim that the following linear program computes a leader's optimal commitment that incentivizes an action $b$ as a best response. 
    Note that it has $O(|A|)$ variables and $O(|A|+|B|)$ constraints. 
    Since linear programs can be solved in time polynomial in the number of variables and constraints, the correctness of the linear program suffices to prove the theorem. 

\begin{align*}
    &\text{Maximize} \quad \sum_{a \in A} p_a  u_1(a, b) - P \\
    &\text{Subject to:} \\
    &P \geq 0 \\
    &\sum_{a \in A} p_a = 1 \\
    &p_a \geq 0 \hspace{102pt}\text{for all} \quad a \in A  \\
    &\sum_{a \in A} p_a u_2(a, b) + P \geq \sum_{a \in A} p_a u_2(a, b') \hspace{5pt} \text{for all } b' \in B\setminus \{b\}
\end{align*}

    Variables $p_i$ represent the probability with which the leader plays each action $a_i$, and $P$ represents the payment from the leader to follower when the follower plays $b$. 
    The utilities $u_i(\cdot)$ and action sets $A$ and $B$ are parameters of the problem instance and so are constants from the perspective of the LP. 

    By \Cref{lemma:n=2_mixed_leader-types_payments-only-depend-on-follower-action} (which applies to a generalization of the present setting), it suffices to consider payment functions that depend only on the follower's action. 
    Since the leader seeks to incentivize action $b$, there is no benefit to making non-zero payments when the follower plays any other action $b'\neq b$. 
    Hence, the leader's choice of payment function amounts to choosing $P(b)$, which is what variable $P$ represents.

    The objective is the leader's expected utility assuming $b$ is played. 
    The first constraint ensures that the payment amount is non-negative.
    The next two sets of constraints ensure that the $p_i$ correspond to a valid mixture over actions and the final set of constraints ensures that $b$ is a best response to the leader's commitment.\footnote{Technically, this only ensures that $b$ is a best response before tie-breaking in favor of the leader.
    However, if another action $b'$ resulted in equal utility for the follower and higher utility for the leader, it would be possible to incentivize $b'$ over $b$ with the mixture \{$p_i$\} as in the current LP and a zero payment function. 
    The LP finding the optimal commitment incentivizing $b'$ optimizes over a region including that commitment, and would find a commitment at least as good. 
    Hence, our approach is still correct. 
    }
    These constraints are clearly necessary and sufficient for the $\{p_a\}$ and $P$ to correspond to a valid commitment incentivizing $b$. 

    This proof adapts the approach of \citet{conitzer_optimal_strategy_to_commit} for computing optimal commitments in the version of the setting without payments.
\end{proof}

\subsection{More than Two Players, Single Commitment}

\thmThreePlayerSingleCommitmentHardness*
\begin{proof}
We reduce from \textsc{Balanced Complete Bipartite Subgraph} in a way that's similar to but simpler than the reduction from \textsc{CLIQUE} to \textsc{Best Nash} in Section 3.1 of Gilboa and Zemel (1981) 
\cite{gilboa1989nash-complexity}.
Let $G=(V,E)$ be a bipartite graph with partite sets $V_1$ and $V_2$. %
Then we define the following three-player normal-form game. 
Player 1 has only one action. Because of this, we'll leave Player 1's actions out of strategy profiles. 
Player 2's action space is  $\{C\} \times V_1 \cup \{E\} \times V_2$ and Player 3's action space is $\{C\} \times V_2 \cup \{E\} \times V_1$ (where $C$ stands for ``cooperate'' and $E$ stands for ``exploit'').

The utilities are as follows:
\begin{itemize}
    \item If nodes $v_1\in V_1$ and $v_2\in V_2$ are adjacent, then
    \begin{eqnarray*}
    u_1((C,v_1),(C,v_2)) &=& u_2((C,v_1),(C,v_2))\\
    &=& u_3((C,v_1),(C,v_2))\\
    &=&1.
    \end{eqnarray*}
    \item If nodes $v_1\in V_1$ and $v_2\in V_2$ are \textit{not} adjacent, then
    \begin{eqnarray*}
    u_1((C,v_1),(C,v_2)) &=& u_2((C,v_1),(C,v_2))\\
    &=& u_3((C,v_1),(C,v_2))\\
    &=& 0.
    \end{eqnarray*}
    \item All other utilities of Player 1 are $0$. 
    \item For $v_2\in V_2$, $u_{2/3}((E,v_2),(C,v_2))=(k,-k-1)$.%
    \item For $v_1\in V_1$, $u_{2/3}((C,v_1), (E,v_1))=(-k-1,k)$.
    \item For $v_1\neq v_1'$ and $v_2\neq v_2'$, $$u_{2/3}((E,v_2'),(C,v_2))= u_{2/3}((C,v_1),(E,v_1'))=(0,0).$$
    \item For any $v_1\in V_1,v_2\in V_2$,
    $$u_{2/3}((E,v_2),(E,v_1))=(0,0).$$
\end{itemize}
So, intuitively, Player $2$ and Player $3$ can either play ``cooperatively'' by playing adjacent nodes in $V_1$ and $V_2$, respectively, or they can try to play exploitatively by trying to guess the other player's node.
If both were guaranteed to play cooperatively, the game would be a coordination game in which all players benefit when adjacent vertices are played.
However, since each player needs to avoid being exploited, neither player can play any of the cooperative actions with probability greater than $1/k$.
The payoff matrix of the game is given in \Cref{table:reduction-n-greater-2-single}.

\begin{table}[t]
	\begin{center}
    \setlength{\extrarowheight}{2pt}
    \begin{tabular}{cc|C{2cm}|C{2cm}|}
      & \multicolumn{1}{c}{} & \multicolumn{2}{c}{P3}\\
      & \multicolumn{1}{c}{}  & \multicolumn{1}{c}{$(C,v_2)$} & \multicolumn{1}{c}{$(E, v_1')$} \\\cline{3-4}
      \multirow{4}*{P2} & $(C,v_1)$ & $(1,1,1)$ if $(v_1,v_2)\in E$; else $(0,0,0)$ & $(0,-k-1,k)$ if $v_1'=v_1$; else $(0,0,0)$ \\\cline{3-4}
      & $(E,v_2')$ & $(0,k,-k-1)$ if $v_2'=v_2$; else $(0,0,0)$ & $0,0,0$  \\\cline{3-4}
    \end{tabular}
    \end{center}
    \caption{The game constructed for the proof of \Cref{thm:3-player-single-commitment-hardness}. The utilities of all three players are given in order.}
    \label{table:reduction-n-greater-2-single}
  \end{table}
  
Now we show that the following two statements are equivalent:
\begin{itemize}
    \item There's a balanced complete bipartite subgraph of size $k$ in the graph $G$.
    \item Player 1 has a commitment to payments s.t.\ Player 1's utility is at least $1$ in some Nash equilibrium of the induced game between Player 2 and Player 3.
\end{itemize}

$\Rightarrow$: Suppose $V_1'\subseteq V_1,V_2'\subseteq V_2$ is a balanced complete bipartite subgraph of size $k$. It is easy to see that if Player 1 makes no payments, uniform mixing by Player $2$ and Player $3$ over $\{(C,v_i)\mid v_i\in V_{i}'\}$ is a Nash equilibrium and that all players receive a utility of $1$ in this Nash equilibrium. (Deviating to some $(E,v_i)$ for $v_i \in V_i'$ yields a utility of exactly $1$ for the deviating player.)

$\Leftarrow$: Let $P_1$ be payments and $\sigma_2,\sigma_3$ be mixtures for Players 2 and 3 such that $(\sigma_2,\sigma_3)$ is a Nash equilibrium of $G[P_1]$ in which Player 1 receives utility at least $1$ after payments.

Notice first that the only way for P1 to receive utility $1$ is for $\mathrm{supp}(\sigma_2)$ and $\mathrm{supp}(\sigma_3)$ to consist only of the $(C,*)$ actions and for each vertex in $\mathrm{supp}(\sigma_2)$ to be adjacent to each vertex in $\mathrm{supp}(\sigma_3)$. Furthermore, Player 1 cannot make any on-equilibrium payments.

Second, it is easy to see that neither player can play any of these actions $(C,v)$ with probability greater than $1/k$, because if they were to do so, the other player could better-respond by playing the corresponding $(E,v)$ action. 
(Player 1 cannot make on equilibrium payments and her off-equilibrium payments cannot decrease the incentive to deviate to $(E,v)$.) 
 
It follows immediately that $\mathrm{supp}(\sigma_2)$ and $\mathrm{supp}(\sigma_3)$ induce a balanced complete bipartite subgraph graph of size at least $k$.
\end{proof}

\subsection{More than Two Players, Sequential Commitment}

\thmThreePlayerIteratedPureLP*
\begin{proof}
    We give a LP-based algorithm that computes, for any outcome $a$, the leader's utility-maximizing commitment that implements $a$ and the utility it achieves. 
    Running this algorithm for each of the $|A|$ action profiles and then maximizing over them computes the leader's optimal strategy. 

    To implement $a$, the leader must commit to $a_1$ and to payments that implement $a$. That is, she must ensure that Player $2$ commits to action $a_2$ and to payments that implement $a$ (i.e., that incentivize Player $3$ to commit to $a_3$).
    
    Consider the following linear program. 
    \begin{align*}
    &\text{Maximize} \quad  u_1(a) - t_{1,2}(a) - t_{1,3}(a)\\
    &\text{Subject to:}\\
    & t_{i,j}(a) \geq 0  \quad\text{for all } (i,j) \in \{(1,2), (1,3), (2,3)\} & (1)\\
    & u_3(a) + t_{1,3}(a) + t_{2,3}(a) \geq u_3(a_1, a_2, a_3') \quad \text{for all } a_3' \in A_3 \setminus \{a_3\} & (2)\\
    & u_2(a) + t_{1,2}(a) + u_3(a) + t_{1,3}(a)  &(3) \\ 
    &\quad \geq u_2(a_1, a_2, a_3') + u_3(a_1, a_2, a_3') \quad \text{for all } a_3' \in A_3 \setminus \{a_3\}  \\
    & u_2(a) + t_{1,2}(a) - t_{2,3}(a) \geq \min_{a_3''} u_2(a_1, a_2', a_3'')  \text{ for all } a_2' \in A_2 \setminus \{a_2\} &(4)
    \end{align*}

    The variables are the $t_{i,j}(a)$, all of which are real valued scalars.
    The utilities $u_i(\cdot)$ and sets $A_i$ are parameters of the game instance and hence constants from the perspective of the LP. 
    (The action profile $a$ is also an input to the LP.)
    It's easy to verify that the program in indeed linear.
    Note that the LP has three variables and $O(|A_2|+|A_3|) \in O(|A|)$ constraints, both of which are polynomial in the size of the game instance, so can be solved efficiently.  

    We now seek to show that the LP can indeed be used to compute the leader's optimal commitment that implements outcome $a$. 
    The variables $t_{i,j}(a)$ represent the payment from Player $i$ to $j$ when action profile $a$ is played.
    We argue that satisfying the constraints is both necessary and sufficient for the existence of a payment function $P_1$ with $P_{1,2}(a) = t_{1,2}(a)$ and $P_{1,3}(a) = t_{1,3}(a)$ that implements $a$ (and that this $P_1$ can be computed efficiently). %

    First, constraint (1) simply ensures the payment variables are non-negative, which is necessary and sufficient for them to be valid.

    \textbf{Necessary:} 
    To implement $a$, the leader needs to commit to $(a_1, P_1)$ such that there exists a  commitment $(a_2, P_{2})$ for Player $2$ such that all three of the following hold:
    \begin{itemize}
        \item Player $3$ prefers playing $a_3$ to any $a_3'$ in $G[(a_1, a_2), (P_1, P_2)]$.
        \item Player $2$ prefers $(a_2, P_{2})$ to any $(a_2', P_{2}')$ where $a_2' \neq a_2$.
        \item Player $2$ prefers $(a_2, P_{2})$ to any $(a_2, P_{2}')$ where $a_3$ is not a best response for Player $3$ in $G[(a_1, a_2), (P_1, P_2)]$.
    \end{itemize}
    
    We'll show that the constraints $(2-4)$ are necessary for this.

    First, observe that the utility Player $3$ can achieve by playing some $a_3'$ after the first two players commit to $a_1$ and $a_2$ is at least $u_3(a_1, a_2, a_3')$ regardless of any payments.  
    Hence, for Player $3$ to instead commit to $a_3$, her utility for doing so $u_3(a) + t_{1,3}(a) + t_{2,3}(a)$ must be at least $u_3(a_1, a_2, a_3')$.
    This is exactly what constraint (2) requires, so it is necessary. 
    
    Now, notice that if Player $2$ commits to any $a_2'$ and the uniformly zero payment function, he is guaranteed a utility of at least $\min_{a_3'} u_2(a_1, a_2', a_3')$ (regardless of Player $1$'s payments).  
    Therefore, Player $2$'s utility when $a$ obtains, $u_2(a) + t_{1,2}(a) - t_{2,3}(a)$, must be at least that amount. 
    Constraint (4) requires exactly that, so is necessary. 

    Finally, we show constraint (3) in the LP must be safisfied for the third bullet above to hold. 
    Player $2$'s maximum utility for playing $a_2$ and implementing any $\hat a_3$ can be computed via the simple DP-like approach from the two-player case (\Cref{Thm:2-player-pure-commitment}).
    Let $m = \max_{\hat a_3} u_3(a_1, a_2, \hat a_3) + P_{1,3}(a_1, a_2, \hat a_3)$, and let $\hat a_3$ be an action realizing that maximum.

    For Player $2$ to implement $a_3$, he must pay Player $3$ at least $m - [u_3(a) + t_{1,3}(a)]$ when $a$ obtains so that Player $3$ prefers playing $a_3$ to $\hat a_3$, leaving him with utility at most $u_2(a) + t_{1,2}(a) - m + u_3(a) + t_{1,3}(a)$.
    Player $2$ can also implement any $a_3'$ by paying $m - [u_3(a_1, a_2, a_3') + t_{1,3}(a_1 a_2 a_3')] \leq m - u_3(a_1 a_2 a_3')$ to Player $3$ when $(a_1, a_2, a_3')$ obtains and $0$ otherwise. %
    This gives him utility $u_2(a_1, a_2, a_3') + t_{1,2}(a_1, a_2, a_3') - m + u_3(a_1, a_2, a_3') \geq u_2(a_1, a_2, a_3') - m + u_3(a_1, a_2, a_3')$.
    Since Player $2$ must prefer implementing $a$ to implementing $(a_1, a_2, a_3')$, it is necessary that $u_2(a) + t_{1,2}(a) - m + u_3(a) + t_{1,3}(a) \geq u_2(a_1, a_2, a_3') - m + u_3(a_1, a_2, a_3')$, which is identical to constraint (3) after cancelling out the $-m$ terms.

    \textbf{Sufficient:}
    Suppose the LP finds an assignment of variables $t_{i,j}(a)$ achieving some value of the objective. 
    Then we claim that the leader can implement $a$ by committing to action $a_1$ and to a payment function $P_1$ whose nonzero values are as follows: 
    \begin{itemize}
        \item $P_{1,2}(a) = t_{1,2}(a)$
        \item $P_{1,3}(a) = t_{1,3}(a)$
        \item $P_{1,3}(a_1 a_2' m(a_2')) = M$
    \end{itemize}
    where $m(a_2') \in \arg\min_{a_3'} u_2(a_1, a_2', a_3')$ and $M$ is a very large constant $M \geq \sum_{i \in \{2,3\}} \left[ \max_{a'} u_i(a') - \min_{a'} u_i(a')\right]$. Let $M_2 = \\ \max_{a'} u_2(a') - \min_{a'} u_2(a')$ and $M_3 = \max_{a'} u_3(a') - \min_{a'} u_3(a')$, so $M \geq M_2 + M_3$. Intuitively, $M_i$ is the difference between Player $i$'s largest and smallest possible utilities.

    Define $t_{2,3}^*(a)$ to be the minimum value of $t_{2,3}(a)$ which is compatible with the constraints, making the lower bound from (2) tight. 
    We claim that $(a_2, P_2)$ is optimal for Player $2$ in $G[(a_1, P_1)]$, where $P_{2,3}(a') = \begin{cases}
        t_{2,3}^*(a') &\text{if } a'=a \\
        0 &\text{if } a'\neq a 
    \end{cases}$, and that this commitment incentivizes Player $3$ to play $a_3$.
    
    The second part of the claim follows immediately from constraint (2) since the left hand side is Player $3$'s utility for committing to $a_3$ in $G[a_1, a_2, P_1, P_2]$ and the right hand side is Player $3$'s utility for committing to some $a_3'\neq a_3$ in the same game.\footnote{Technically, this is only true up to tiebreaking if the inequality is tight. However, as we now argue, this doesn't cause any problems.   
    
    If constraint (2) is tight for some alternative action $a_3'$, Player $3$ will tie-break in favor of Player $2$ (and then in favor of the leader if Player $2$ is also indifferent). 
    Intuitively, Player $2$ also weakly prefers $a_3$ to $a_3'$ because if Player $3$ is indifferent between $a_3$ and $a_3'$, Player $2$ could have implemented $a_3'$ via an arbitrarily small payment, but did not choose to do so. 
    Formally, by constraint (3), we have $u_2(a) + t_{1,2}(a) + u_3(a) + t_{1,3}(a) \geq u_2(a_1, a_2, a_3') + u_3(a_1, a_2, a_3')$. 
    Since Player $3$ is indifferent, $u_3(a) + t_{1,3}(a)+ t_{2,3}(a) = u_3(a_1, a_2, a_3')$ and we can simply the previous expression to $u_2(a) + t_{1,2}(a) - t_{2,3}(a) \geq u_2(a_1, a_2, a_3')$, as desired. 

    If this inequality is strict, Player $3$ will play $a_3$, as desired. 
    If not, Player $3$ will tiebreak in favor of Player $1$, which is is fine for the same reason in the proof of \Cref{Thm:2-player-mixed-strat}: 
    The commitment found actually implements $(a_1, a_2, a_3')$, and the LP that finds the optimal implementation of that outcome will have a higher leader payoff. 
    }
    (There is no reason for Player $3$ to commit to payments for any outcome that actually obtains because all other players have already committed to actions). 

    Now, we claim that it is optimal for Player $2$ to commit to $a_2$ and $P_2$ in $G[a_1, P_1]$. 
    This results in utility $u_2(a) + t_{1,2}(a) - t_{2,3}^*(a)$ after Player $3$ plays $a_3$.
    
    First, consider an alternative commitment to an action $a_2'$ (along with possible payments).
    Player $3$'s utility for committing to $m(a_2')$ would be at least $$u_3(a_1, a_2', m(a_2')) + M \geq M_2 + \max_{a'} u_3(a').$$ 
    Hence, incentivizing Player $3$ to play some other $a_3'$ would require payments of at least $$M_2 + \max_{a'} u_3(a') - u_3(a_1, a_2', a_3') \geq M_2.$$ 
    But making a payment of $M_2$ would leave Player $2$ with utility at most $\min_{a'} u_2(a')$, and so would not be worthwhile. 
    Hence, Player $2$'s optimal utility when committing to play $a_2'$ is $u_2(a_1, a_2', m(a_2')) = \min_{a_3'} u_2(a_1, a_2', a_3')$, which is no more than her utility for following the equilibrium of $u_2(a) + t_{1,2}(a) - t_{2,3}^*(a)$ by constraint (4). 
    Therefore, it is optimal for Player $2$ to commit to $a_2$. 

    Now, we claim it is optimal for Player $2$ to commit to implement $a_3$. 
    Consider some other $a_3' \neq a_3$.
    Suppose $u_3(a) + t_{1,3}(a) - u_3(a_1, a_2, a_3') = \delta$. Let's first consider the case where $\delta \geq 0$.%
    Then Player $2$ would need to pay $P_{2,3}(a_1, a_2, a_3') \geq \delta$ to incentivize Player $3$ to play $a_3'$. 
    But this would give Player $2$ utility $u_2(a_1, a_2, a_3') - \delta$, which is less than $u_2(a) + t_{1,2}(a)$ by constraint (4). 
    Now, suppose $\delta < 0$.
    Then Player $2$ could incentivize Player $3$ to play $a_3$ by paying $P_{2,3}(a_1, a_2, a_3) = -\delta$. 
    This would give Player $2$ utility $u_2(a) + t_{1,2}(a) + \delta = u_2(a) + t_{1,2}(a) + u_3(a) + t_{1,3}(a) - u_3(a_1, a_2, a_3') \geq u_2(a_1, a_2, a_3')$ by constraint (4). 
    Hence Player $2$ always prefer to implement $a_3$ over any $a_3'$. 
\end{proof}

    One may notice that $t_{2,3}(a)$ is lower bounded by constraint (2) and upper bounded by constraint (4), but unconstrained within those bounds. 
    Intuitively, the lower bound ensures that the payment is high enough that Player $3$ is incentivized to play $a_3$, and the upper bound ensures that Player $2$ prefers to implement $a$ (by committing to $t_{2,3}(a)$) over playing some other $a_2'$.
    These constraints ensure that the leader's commitment incentivizes Player $2$ to implement $a$, which is sufficient to solve the leader's problem.
    However, an arbitrary value of $t_{2,3}(a)$ from a solution to the LP won't correspond to the actual amount Player $2$ would commit to pay -- if Player $2$ implements $a$, he'll do so by committing to pay the minimum value which is sufficient to do so. 
    Therefore, for ease of analysis, we'll define $t_{2,3}^*(a)$ to be the minimum value of $t_{2,3}(a)$ which is compatible with the constraints, making the lower bound from (2) tight.\footnote{Another possible approach would be to have a lexicographic objective for the LP, where the first is the current objective and the second is to minimize $t_{2,3}(a)$. 
    While this would also suffice, it introduces a slight inelegance in that not all optimal commitments from the leader would be optimal solutions to the LP. 
    For instance, if there are two optimal commitments for the leader and only one requires Player $2$ to make payments, then only the other will be an optimal solution to the LP, even if the first is actually better for Player $2$.}

\thmIterativeMixedCommitmentHard*

\begin{proof}
We reduce from \textsc{Balanced-Vertex-Cover}. Let $(V,E)$ be a graph. We construct a three-player game as follows. 
For each vertex $v$ in the graph, each player has an action corresponding to $v$. We call these actions $a_v, b_v, c_v$ for Players 1, 2 and 3, respectively. In addition, Player 3 has an action $c_e$ for each edge $e$ and one additional action $c_0$.

Let $\epsilon \in (0, \nicefrac{1}{|V|^5}]$.
The players' utilities are given as follows:
\begin{itemize}
    \item For all $a\in A$ and $b\in B$, set $u_1(a,b, c_0)=u_2(a,b,c_0)=\epsilon$. %
    \item For all $a\in A$, $b\in B$, and $c\in C\setminus \{c_0\}$, set $u_1(a,b,c)=u_2(a,b,c)=0$.
    \item For all $v\in V$ and $ b\in B$, set $u_3(a_v,b,c_v)=0$.
    \item For all $v\in V$ and $a\in A$, set $u_3(a,b_v,c_v)=0$.
    \item For all $v\in V$ and for all $a\in A\setminus \{a_v\}$ and $b\in B\setminus \{b_v\}$, set $u_3(a,b,c_v)=\frac{|V|}{|V|-2}$.
    \item For all $e\in E,b\in B$ and both $v\in e$, set $u_3(a_v,b,c_e)=0$.
    \item For all $e\in E,b\in B$ and all $v\notin e$, set $u_3(a_v,b,c_e)=\frac{|V|}{|V|-2}$.
    \item For all $a\in A,b\in B$, set $u_3(a,b,c_0)=1$.
\end{itemize}
The problem is to decide whether the leader has a commitment to payments and a mixture over actions achieving a utility of at least $k$, where $k$ can be picked arbitrarily from $(0,\epsilon]$.

This construction is almost identical to that of \citet{conitzer_optimal_strategy_to_commit}: The only difference is that Player 1 and 2's utilities are scaled down by a factor of $\epsilon$. 
This change makes Players 1 and 2's ability to commit to payments inconsequential: they can only commit very small amounts that, as we will show, can't influence Player 3's optimal strategy.

We first give some intuition for this game.
Note first that Players 1 and 2 have the same utility function. Also, note that Players 1 and 2 can only achieve nonzero utility if Player 3 plays $c_0$. Thus, the question of whether Player 1 can achieve utility at least$k$ is equivalent to the question of whether Players 1 and 2 can, via their commitments to mixtures and payments, (weakly) incentivize Player 3 to play $c_0$. Player 3 meanwhile, has various different actions that may give high utility. Roughly, the construction is such that Player 3 can earn a utility greater than $1$ if and only if Players 1 and 2 fail to play mixtures corresponding to a balanced vertex cover and Player 3 points out a \enquote{gap} in Player 1 and 2's \enquote{claimed} vertex cover.

To show that the reduction is valid, we need to show that a balanced vertex cover of $(V,E)$ exists if and only if Player 1 can achieve a utility of at least $k$ via some commitment to payments and a mixture over actions.

\underline{$\Rightarrow$}: Assume that a balanced vertex cover exists. Consider the action profile where Player 1 commits to mix uniformly over the actions $a_v$ for $v$ in the balanced vertex cover and Player 2 commits to mix uniformly over the actions $b_v$ for $v$ \textit{not} in the vertex cover. Finally, Player 3 plays $c_0$. Nobody commits to any payments.

We need to show that Player $2$ and $3$ play best responses. 
First, we verify that Player $3$ has no profitable deviations. Since Players $1$ and $2$ have already committed, there is no point in committing to payments. Thus, we only need to show that Player $3$ cannot achieve a utility greater than $1$ by playing an action other than $c_0$. This is done in the same way as in \citeauthor{conitzer_optimal_strategy_to_commit}'s proof:
For each vertex $v$, Player $3$'s expected utility for playing $c_v$ is
$$
\left( 1-\frac{2}{|V|}\right)\frac{|V|}{|V|-2}=1
$$
because for each $v$, there is a $1-\frac{2}{|V|}$ chance that $a_v$ nor $b_v$ is played.
For each edge $e$, Player $3$'s expected utility for playing $c_e$ is at most
$$
\left( 1- \frac{2}{|V|}\right) \frac{|V|}{|V|-2}=1,
$$
because for each $e=\{v_1,v_2\}$, Player $1$ plays at least one of $v_1$ and $v_2$ with probability $ \frac{2}{|V|}$.

It remains to show that Player 2 cannot profitably deviate. Note that Player 2 achieves his maximum possible utility before payments. So, if Player 2 can profitably deviate it must be by receiving payments from one of the other players. However, Player 1 has already committed not to make any payments. Furthermore, Player 3 never commits to payments that are actually made, because when Player 3 commits to payments, Players 1 and 2 have already committed to actions. Therefore, Player 2 cannot achieve a utility higher than $\epsilon$.

\underline{$\Leftarrow$}: We now show that if Player 1 can achieve positive utility, then a balanced vertex cover of the graph exists. 
To achieve positive utility, Players 1 and 2 need to induce Player 3 to play $c_0$. Players 1 and 2 can pay Player 3 for playing $c_0$.
However, if Players $1$ and $2$ pay more than $2\epsilon$ in total to Player $3$, then at least one of them would have negative utility if Player $3$ played $c_0$. 
That player could then profitably deviate by not committing to any payments. 
Therefore, including payments, Player 3's utility for playing $c_0$ is at most $1+2\epsilon$. 
This means that Players 1 and 2's strategies must ensure that all actions other than $c_0$ give utility at most $1+2\epsilon$. 
The rest of the proof will show the following: If all of Player 3's actions give utility at most $1+2\epsilon$, then we can use Player 1's strategy to construct a balanced vertex cover for the underlying graph. 
The proof is structurally similar to that of \citeauthor{conitzer_optimal_strategy_to_commit}. 
The main difference is that we have to deal with the fact that Players 1 and 2 could try to pay Player $3$ to play $c_0$. We have to show that the payments that Players 1 and 2 can \enquote{afford} are too small to get Player $3$ to play $c_0$.

First, consider Player 3's actions $c_v$ for vertices $v$. For Player 3's utility for $c_v$ to be at most $1+2\epsilon$, it must be the case that 
$$
(1-p_v) \frac{|V|}{|V|-2}\leq 1+2\epsilon,
$$
where $p_v$ is the probability that at least one of Players 1 and 2 plays $a_v$ or $b_v$. Solving for $p_v$, we obtain that
\begin{eqnarray*}
p_v &\geq& 1- \frac{(|V|-2)(1+2\epsilon)}{|V|}\\
    &=& 1-\frac{|V|-2 + 2\epsilon|V| - 4\epsilon}{|V|}\\
    &=& \frac{2}{|V|}-2\epsilon+\frac{2\epsilon}{|V|}\\
    &\geq& \frac{2}{|V|} -2\epsilon.
\end{eqnarray*}

Now we would like to show that for each vertex $v$, one of Players 1 and 2 plays their action corresponding to $v$ ($a_v$ or $b_v$) with probability close to $\frac{2}{|V|}$ and the other plays it probability close to $0$. %
Intuitively, this is because they need to use their overall probability mass of $2$ to cover $|V|$ nodes with probability mass almost $\frac{2}{|V|}$ each.
If they play both $a_v$ and $b_v$ with positive probability, they \enquote{waste} some probability mass, because there is some chance that they play $a_v$ and $b_v$ simultaneously (which gives no benefit over having just one of them play $a_v$ or $b_v$).

First, note that each $v$ must be covered with probability at least $\frac{2}{|V|}-2\epsilon$, requiring a total probability mass of at least $|V|\left( \frac{2}{|V|}-2\epsilon \right)$.
Hence, Players $1$ and $2$ can play strategies corresponding to the same vertex with probability at most $2-|V|\left( \frac{2}{|V|}-2\epsilon \right)=2|V|\epsilon$, and so in particular for any $v$, the probability that both Player $1$ plays $a_v$ \textit{and} Player $2$ plays $b_v$ can be at most $2|V|\epsilon$.

Notice that for each vertex $v$, at least one of Player 1 and 2 plays $a_v$/$b_v$ with probability $\geq \frac{1}{|V|}-\epsilon$. We now show that the other of the two players plays the corresponding action with low probability. 
If the other player plays it with probability $\delta$, then there is a probability of at least
$$
\delta \left(\frac{1}{|V|} -\epsilon \right)
$$
that both Players 1 and 2 play the action corresponding to $v$. Because only $2|V|\epsilon$ can be wasted, we get that
$$
\delta \left(\frac{1}{|V|} -\epsilon \right) \leq 2|V|\epsilon,
$$
which is equivalent to
$$
\delta \leq 2|V|\epsilon / \left(\frac{1}{|V|} -\epsilon \right)\leq 2|V|\epsilon / \left(\frac{1}{2|V|} \right) = 4|V|^2\epsilon \leq \frac{4}{|V|^3}.
$$
Call this $\bar \delta$.
It follows immediately that the player who plays the action corresponding to $v$ with probability $\geq \frac{1}{|V|}-\epsilon$ must actually play $v$ with probability at least $ \frac{2}{|V|}-\epsilon - \bar \delta$ in order to ensure that $v$ is played with probability at least $ \frac{2}{|V|}-\epsilon$.

Next we show that Players 1 and 2 each play exactly half of the actions with probability at least $ \frac{2}{|V|} -\epsilon-\bar\delta$ and the other half with probability $\leq \bar \delta$. For this, it is enough to show that neither player can play \textit{more} than half of vertices with probability at least $ \frac{2}{|V|} -\epsilon-\bar\delta$. Notice that to play $|V|/2$ vertices with probability $\frac{2}{|V|}-\epsilon - \bar \delta$, a player spends
$$
\frac{|V|}{2}\left(\frac{2}{|V|}-\epsilon - \bar \delta\right) = 1- \frac{|V|(\epsilon + \bar \delta)}{2} 
$$
of his or her available probability mass. Thus, the player has only $\frac{|V|(\epsilon + \bar \delta)}{2}$ left to distribute among other vertices. The following shows that this is less than $\frac{2}{|V|}-\epsilon-\bar\delta$: 
\begin{eqnarray*}
\frac{|V|(\epsilon + \bar \delta)}{2} &=& \frac{|V|\left(\frac{1}{|V|^5} + \frac{4}{|V|^3}\right)}{2}\\
&=& \frac{1}{2|V|^4} + \frac{2}{|V|^2}\\
&\leq& \frac{3}{|V|^2}\\
&\leq & \frac{1}{|V|}\quad (\text{using }|V|\geq 3)\\
&<& \frac{2}{|V|}-\frac{1}{|V|^5}-\frac{4}{|V|^3}\\
&\leq& \frac{2}{|V|}-\epsilon-\bar\delta.
\end{eqnarray*}

Finally, we show that if Player 3 prefers $c_0$, then the vertices $v$ for which Player 1 plays $a_v$ with probability at least $ \frac{2}{|V|}-\epsilon -\bar \delta$ form a vertex cover of the graph $G$. We prove this by contraposition: we prove that if they do not form a valid vertex cover, then Player 3 will prefer some action $c_e$ over $c_0$. Let $e$ be an edge such that Player 1 plays both vertices of $e$ with probability at most $\bar \delta$, so that with probability at least $ 1-2\bar\delta$ Player 1 plays neither of the vertices of $e$.
Then we can show that Player 3 prefers $c_e$ over $c_0$ as follows:
\begin{eqnarray*}
(1-2\bar\delta)\frac{|V|}{|V|-2} &\geq& \left(1-\frac{8}{|V|^3} \right)\frac{|V|}{|V|-2}\\
&=& \frac{|V|}{|V|-2} - \frac{8}{(|V|-2)|V|^2}\\
&=& 1+\frac{2}{|V|-2} - \frac{8}{(|V|-2)|V|^2}\\
&\geq & 1+\frac{2}{|V|-2} - \frac{1}{|V|}\text{ using }|V|\geq 4\\
&\geq & 1+\frac{1}{|V|}\\
&\geq & 1+2\epsilon \text{ using }|V|\geq 2.
\end{eqnarray*}

Recall that $1+2\epsilon$ is an upper bound on the utility Player $3$ can receive for $c_0$.
Therefore, Player 3 prefers $c_e$ to $c_0$, and hence by contraposition there exists a balanced vertex cover of $G$. 
\end{proof}

\section{Deferred Proofs from Section~\ref{sec:CE}}\label{appendix:CE-proofs}

Recall that \Cref{sec:CE}, we assumed that the leader commits to a distribution $D \in \Delta(A)$. 
She then draws a profile $a \sim D$, plays $a_1$ herself, and sends the recommendation $a_i$ to each other player $i$. 
She also commits to a payment function $P_i: A_i \rightarrow \mathbb{R}_{\geq 0}$, where $P_i(a_i)$ is the amount she pays agent $i$ for following her recommendation when she recommends playing $a_i$. 
We also assume that the leader makes an \textit{incentive compatible} commitment, meaning that all followers are incentivized to follow their recommended action (in a Bayes-Nash equilibrium of the induced game). 

With the following lemma, we essentially show that those assumptions are without loss of generality. 
That is, the type of leader commitment we considered in the main text is no less powerful than one in which the leader commits to a mixture over actions and then to an arbitrary signaling scheme based on her realized action, along with a payment function that can depend on the entire signal profile and outcome. 

\begin{lemma}\label{lemma:CE-signaling-revelation-principal-and-payments-simplification}
    Suppose the leader can commit to a mixture over actions, an arbitrary signaling scheme based on her (realized) action, and a payment function that depends on the realized outcome and all players' signals.
    That is, she commits to a mixture over actions $D_1 \in \Delta(A_1)$ and function $S: |A_1| \rightarrow \Delta([M]^{n-1})$ where $[M] = \{1, 2, \dots, M\}$ for some arbitrarily large natural number $M$, and a payment function $P: A \times [M]^{n-1} \rightarrow \mathbb{R}_{\geq 0}^{n-1}$. 
    Then, she draws an action $a_1 \sim D_1$, plays $a_1$, draws a signal $s \sim S(a_1)$ and sends each player $i$ private signal $s_i$.

    There exists a distribution $D \in \Delta(A)$ and a payment function $\bar P$ with $\bar P_{i}: A_i \times \{0,1\} \rightarrow \mathbb{R}_{\geq 0}$ in which each player receives a payment $\bar P_i(a_i)$ if they follow their recommendation to play $a_i$ and $0$ otherwise such that the distribution over outcomes and expected payoffs are the same under $(D_1, S, P)$ and under $(D, \bar P)$. 
\end{lemma}

\begin{proof}
    We give a revelation principal type argument which we extend to account for payments. 

    Consider an equilibrium induced by the leader commitment $C = (D_1, S, P)$. 
    For each player $i$, let $\sigma_i: [M] \rightarrow \Delta(A_i)$ give the mixture over actions Player $i$ plays after receiving a message $s_i = k$. 
    For convenience, let's define $\sigma_1(s)$ to be the leader's posterior distribution over actions given that she sends signal $s$, so  $\sigma_1(s)(a_1) = \frac{S(a_1)(s)}{\sum_{s' \in [M]^n} S(a_1)(s')}$.

    Define $\hat C = (D_1, \hat S, \hat P)$ to be an alternate leader commitment. 
    Let the message space to each player $i$ consist of pairs in $[M]\times A_i$.
    Define $\hat S(a_1)(s, a_{-1}) = S(a_1)(s) \prod_i \sigma_i(s)(a_i)$, so that intuitively, whenever $S$ sends message $k$ to Player $i$, $\hat S$ sends each message $(k, a_i)$ with the probability $\sigma_i(k)(a_i)$ that Player $i$ would have taken action $a_i$.  
    
    Define $\hat P$ as follows, where $\hat s_i = (\hat k, \hat a_i)$
    $$\hat P_i(\hat s, a) = \begin{cases}
        \E_{s\sim S(s | s_i=\hat k)}\left[P_i((a_i, \sigma_{-i}(s)), s)\right] &\text{ if } a_i = \hat a_i \\
       0 &\text{ if } a_i \neq \hat a_i
    \end{cases}.$$
    Intuitively, whenever a player $i$ receives $(\hat k, \hat a_i)$ and takes action $\hat a_i$, $\hat P$ pays them what they would have expected to be paid under $P$ for taking $\hat a_i$ when they received message $\hat k$ from $S$, where the expectation is given their posterior over $s$ and $\sigma_{-i}(s)$.
    It pays them nothing if they receive some $\hat s = (\hat k, \hat a_i)$ and taking some $a_i' \neq \hat a_i$. 

    First, we claim that playing $\hat a_i$ is a best response to receiving $\hat s = (\hat k, \hat a_i)$ under $\hat C$.
    When Player $i$ receives signal $\hat s = (\hat k, \hat a_i)$ under $C'$, their expected utility for playing $\hat a_i$ is exactly their expected utility for playing $\hat a_i$ when receiving signal $\hat k$ under $C$. 
    Meanwhile, their expected utility for playing some $a_i' \neq \hat a_i$ when receiving $(\hat k, \hat a_i)$ under $\hat C$ is at most their utility for playing $a_i'$ when receiving $\hat k$ under $C$ (since their expected utility for the action is the same and their expected payment is weakly lower).
    Therefore, since $(\hat k, \hat a_i)$ is only sent with positive probability for $\hat a_i \in \supp \sigma_i(\hat k)$ and $\sigma_i(\hat k)$ is a best response to receiving $\hat k$, playing $\hat a_i$ is a best response to receiving $(\hat k, \hat a_i)$ under $\hat C$. 
    Therefore, it is an equilibrium of $\hat C$ for all players to take $\hat a_i$ when receiving message $(\hat k, \hat a_i)$. 

    Note also that the distribution over messages $k$ under $C$ is identical to the marginal distribution over messages $(k, \cdot{})$ under $\hat C$, and the distribution over recommendations $(\cdot{}, a_{-1})$ under $\hat C$ is identical to the distribution over actions $a_{-1}$ under $\sigma_{-1}$ under $C$. 
    Hence, $C$ and $\hat C$ induce the same distribution over outcomes. 
    Finally, observe that the expected payments from Player $1$ to each player $i$ are the same under $C$ and $\hat C$. 

    We'll now define another new signaling scheme $\tilde C = (D_1, \tilde S, \tilde P)$ which essentially combines all sets of messages $\{(k, a_i) : k\in [M]\}$ under $\hat S$ into a single message $a_i$. 
    Define $\tilde S: A_1 \rightarrow \Delta(A_{-1})$ by $$\tilde S(a_1)(a_{-1}) = \sum_{k \in [M]^{n-1}} \hat S(a_1)(k, a_{-1}).$$    

    It will be useful to define $\hat S$ directly as a distribution over signals by taking the expectation over the leader's action as we did with $S$, so let $\hat S(s, a_{-1}) = \sum_{a_1 \in A_1} D_1(a_1)\cdot{} \hat S(a_1)(s, a_{-1})$.
    Now, let's define $\tilde P: A \times A_{-1} \rightarrow \mathbb{R}_{\geq 0}^{n-1}$ by 
    $$\tilde P_i\left( a, a'_{-1}\right)  = \begin{cases} \frac{\sum_{k \in [M]^{n-1}} \hat S\left( k, a'_{-1}\right)  \cdot{} \hat P_i\left( a, \left( k, a'_{-1}\right) \right) }{\sum_{k \in [M]^{n-1}} \hat S\left( k, a'_{-1}\right) } &\text{ if } a_i = a'_i
    \\ 0 &\text{ if } a_i \neq a'_i
    \end{cases}.$$
    Notice that $\tilde P_i$ depends only on the signal $\hat a_i$ to Player $i$ and whether or not Player $i$ plays the recommended action $\hat a_i$, which are the desired properties of $\bar P$ in the lemma statement. 

    We claim that $\tilde C$ induces the same distribution over outcomes and same expected payments as $\hat C$, and hence as $C$. 
    
    First, observe that it is an equilibrium for all Players $i$ to play action $a_i$ after receiving signal $a_i$ under $\tilde C$.
    This is because Player $i$'s utility for doing so is a convex combinations of their utility for playing $a_i$ after receiving the signals $(k, a_i)$ under $\hat C$, for all $k \in [M]$.
    (Specifically, it is the convex combination where the coefficients are the marginal probabilities that Player $i$ receives each $(k, a_i)$ under $\hat S$.)
    Likewise, their utility for playing any other action $a_i'$ after receiving signal $a_i$ is an analogous convex combination.
    Since $a_i$ is a best response to any $(k, a_i)$, it is also a best response to a convex combination of such signals, and hence to the signal $a_i$ from $\tilde C$. 

    Therefore, $\tilde C$ and $\hat C$ induce the same distribution over outcomes: whenever $\hat S(a_1)$ sends some $(k, a_{-1})$ and induces outcome $(a_1, a_{-1})$, $\tilde S(a_1)$ sends $(a_{-1})$ and induces the same outcome. 
    They also induce the same expected payments because, whenever Player $i$ plays $a_i$ under $\tilde S$, their payment from $\tilde P$ is their expected payment (over messages $k$) from $\hat P$ given that they received signal $(k, a_i)$ from $\hat S$. 

    Finally, observe that if we let $D$ be the distribution over actions defined by $D(a) = D_1(a_1) \tilde S(a_1)(a_{-1})$, $(D, \tilde P)$ satisfies the requirements of the theorem. 
    \end{proof}

\thmCorrelatedMixedActionNormalForm*

\begin{proof}
    We give a linear program which directly computes the leader's optimal commitment.  
    Note that our LP has $O(|A|)$ variables and $O(|A|^2)$ constraints, each of which is polynomial in the input size. 
    Hence, the LP can be solved in polynomial time, and the result follows from the correctness of the LP, which we now show. 

\begin{align*}
    &\text{Maximize} \quad \sum_{a \in A} p_a u_1(a) - \sum_{i=2}^n \sum_{a_i \in A_i} t_i(a_i) \\
    &\text{Subject to:} \\
    &p_a \geq 0 \quad \text{for all } a \in A\\
    &\sum_{a \in A} p_a = 1 \\
    &t_i(a_i) \geq 0 \quad \text{for all } i \in \{2, \dots, n\} \text{ and } a_i \in A_i \\
    &\sum_{a_{-i}} p_{(a_i, a_{-i})} u_i(a_i, a_{-i}) + t_i(a_i) \geq \sum_{a_{-i}} p_{(a_i, a_{-i})} u_i(a_i', a_{-i}) \\ &\quad \text{for all }  i \in \{2, \dots, n\}, a_i \in A_i, \text{ and } a_i' \in A_i\setminus \{a_i\}
\end{align*}

    The variables are the $p_a$ and $t_i(a_i)$, all of which are real valued scalars.
    The utilities $u_i(\cdot)$ and sets $A_i$ are parameters of the game instance and hence constants from the perspective of the LP. 
    It's easy to verify that the program is indeed linear.

    The variables $p_a$ represent the probability that action profile $a$ is played, and the first two sets of constraints ensure that they form a valid probability distribution. 
    The variables $t_i(a_i)$ represent the product $\sum_{a_{-i}}p(a_i, a_{-i}) \cdot{} P_i(a_i)$ for each follower $i$ and action $a_i$, where $\sum_{a_{-i}}p(a_i, a_{-i})$ is the probability that action $a_i$ is recommended and $P_i(a_i)$ is payment from the leader when follower $i$ follows such a recommendation. 
    That is, $t_i(a_i)$ is the ex ante expected payment (over the action distribution) from the leader to follower $i$ for taking action $a_i$.
    Of course, the standard payment function for the leader's commitment $(D, P)$ can be recovered from the LP variables: $P_i(a_i) = t_i(a_i) / \sum_{a_{-i}}p(a_i, a_{-i})$. 

    The third set of constraints ensures these (normalized) payments $t_i(a_i)$ are non-negative.     
    Note that whenever the probability $\sum_{a_{-i}}p(a_i, a_{-i})$ is strictly positive, there is a one to one correspondence between $t_i(a_i)$ and $P_i(a_i)$. 
    When $\sum_{a_{-i}}p(a_i, a_{-i}) = 0$, i.e. $a_i$ is never recommended, the value of $P_i(a_i)$ is meaningless, and the LP would never find a non-zero value for $t_i(a_i)$. 
    This is because the constraint in which $t_i(a_i)$ appears is trivially satisfied because the other two terms are both $0$, but $t_i(a_i)$ also appears in the objective. 
	Hence, the LP searches over all valid $P_i(a_i)$ and, at optimality, finds only valid assignments.

    The second term in the objective corresponds to the total expected payments made by the leader to all followers, and so the objective is the leader's expected utility: her utility for the outcome minus payments made. 
    
    As discussed above, $t_i(a_i)$ represents $\sum_{a_{-i}}p(a_i, a_{-i}) \cdot{} P_i(a_i)$. 
    Therefore, each constraint in the final set can be viewed as a standard incentive constraint, $u_i(a_i, \sigma_{-i}) + P_i(a_i) \geq u_i(a_i', \sigma_{-i})$ (where $\sigma_{-i}$ is the posterior given the recommendation), multiplied by the probability $\sum_{a_{-i}}p(a_i, a_{-i})$ that action $a_i$ is actually recommended. 
    Hence, if $a_i$ is never recommended, the constraint is always satisfied. 
    If $\sum_{a_{-i}} p_{(a_i, a_{-i})} > 0$, then the constraint simplies to a incentive constraint ensuring Player $i$ gets more utility (after payments) for following the recommendation than for deviating.  
    Therefore, the final set of constraints is satisfied if and only if all followers are willing to follow all recomendations they receive with positive probability. 
    
\end{proof}

\thmCorrelatedPureActionNormalForm*

\begin{proof}
    Consider a particular action $a_1^*$ for the leader. 
    We can compute the leader's optimal strategy while committing to the pure action $a_1^*$ with a straightforward modification of the LP from the proof of \Cref{thm:correlated-n-player-mixed-normal-form}. 
    Consider taking that LP and adding the constraint that the leader commits to the pure action $a_1^*$, that is $\sum_{a | a_1 = a_1^*} p_a = 1$. 
    Everything else remains identical, so by the argument in the previous proof, this new LP computes the leader's optimal commitment to actions, payments, and signals, except that her action is now required to be $a_1^*$. 
    
    Observe that the constraint we add is linear in the variables and that the resulting LP still has the same asymptotic number of variables and constraints. 
    Therefore, in polynomial time, we can solve each of the $|A_1|$ LPs corresponding to the leader's actions and maximize over their objective values to find her optimal commitment. 
\end{proof}

\section{Deferred Proofs from Section~\ref{sec:Bayesian}}\label{appendix:Bayesian-proofs}

\subsection{Follower Types}

\probUnitDemandItemPricing*

\begin{problem*}[\textsc{Unit Demand Pricing for Uniform Budgets}]\label{prob:uniform-budget-min-buying-appendix}
    Given a set of items $M$ and a distribution $D$ over the finite set $C$ of customer types $c$. Each type $c$ is a pair $(b_c, S_c)$ where $b_c \in \mathbb{R}^{\geq 0}$ is a budget and $S_c\subseteq M$ is a subset of items that the customer is interested in. 
    An item pricing is a function $p: M \rightarrow \mathbb{R}^{\geq 0}$
    Let $A(p, c) = \{i \in S_c | p(i) \leq b_c\}$ be the set of items type $c$ is interested in which are below his budget. 
    The problem is to decide whether there exists an item pricing $p$ such that 
    $$Rev(p) = \sum_{c \in C | A(p, c) \neq \emptyset } D(c) \min_{i \in A(p, c)} p(i) \geq K$$
\end{problem*}

\begin{theorem*}[Theorem 5 of \cite{auctions-hardness}]
    \textsc{Unit Demand Pricing for Uniform Budgets} is hard to approximate within $O(|M|^\eps)$ for some $\eps > 0$ if \NP{}
    $\not\subseteq \cap_{\delta > 0} \BPTIME(2^{O(n^\delta)})$.
\end{theorem*}

\thmTwoPlayerFollowerTypesHardness*

\begin{proof}
    Consider an arbitrary instance $U$ of \textsc{Unit Demand Item Pricing} with $m$ items, value distribution $D$, and revenue threshold $K$.
    We'll reduce $U$ to the problem of computing the leader's optimal commitment to a payment function and (mixture over) actions in the two-player Bayesian game $G$ which we now define.

    In $G$, the leader has a single action $s$. 
    The follower has actions $t_1 \dots t_m$ corresponding to the $m$ items as well as an additional action $t_0$. 
    The follower a type $\theta_v$ for each $v$ supported in $D$, and the probability of each follower type $\Pr(\theta_v)$ is $D(v)$.
    Let $Z$ be a constant greater than $\max_{i, v \in \supp(D)} v_i$.
    The leader's utility function is%
    \begin{itemize}
        \item $u_1(s, t_0) = 0$
        \item $u_1(s, t_i) = Z$ for all $i > 0$.
    \end{itemize}
    and follower's utility function is
    \begin{itemize}
        \item $u_2^{\theta_v}(s, t_0) = 0$
        \item $u_2^{\theta_v}(s, t_i) = -Z + v_i$ for all $i >0$.
    \end{itemize}

    Since the leader has only a single action, her payment function is the only meaningful part of her commitment. 
    We claim that there exists a commitment to a payment function $P$ in $G$ that results in utility at least $K$ for the leader if and only if there exists an item-pricing in $U$ giving the principal revenue at least $K$. 
    
    \underline{$\Rightarrow$:} Suppose there exists an item pricing $r \in \mathbb{R}^m_{\geq 0}$ for $U$ which gives the principal expected revenue of at least $K$, that is $Rev(r) = \sum_{v \in A(r)}[ D(v) * r_{i^*(v)}] \geq K$.
    Then we claim that committing to the following payment function $P$ gives the leader utility at least $K$ in $G$:

    \begin{itemize}
        \item $P(s, t_0) = 0$
        \item $P(s, t_i) = Z - r_i \text{ for all } i > 0$.
    \end{itemize}
   
    Recall that $i^*(v)$ is defined as $i^*(v) \in \arg\max_i [v_i - r_i]$, tiebreaking in favor of the largest $r_i$. 
    Observe that playing $t_0$ gives type $\theta_v$ utility $0$, and playing $t_i$ for $i > 0$ gives utility $u_2^{\theta_v}(s, t_i) + P(s, t_i) = (-Z + v_i) + (Z - r_i) = v_i - r_i$.
    Further, the leader's utility when such a $t_i$ is played is $u_1(s, t_i) - P(s, t_i) = Z - (Z - r_i) = r_i$.
    Therefore, $t_{i^*(v)}$ is follower type $\theta_v$'s optimal action among $\{t_i | i > 0 \}$ after accounting for the tie-breaking in favor of the leader. 
    It follows that type $\theta_v$ plays $t_{i^*(v)}$ if $v_{i^*(v)} - r_{i^*(v)} \geq 0$ and plays $t_0$ (receiving utility $0$) otherwise.\footnote{This is consistent with tie-breaking in favor of the leader because the leader's utility for $t_{i^*(v)}$ is $r_i \geq 0$ while and the leader's utility for $t_0$ is 0.}
    Note that these types $\theta_v$ playing $t_{i^*(v)}$ are exactly those in $A(r) = \{ v \in \supp(D) |  v_{i^*(v)} \geq r_{i^*(v)} \}$, and that the leader gets utility $u_1(s, t_0) - P(s, t_0) = 0 - 0= 0$ whenever $v \not\in A(r)$.

    Therefore, the leader's overall expected utility is
    \begin{align*}
        &\sum_{\theta_v : v \in A(r)} \Pr(\theta_v) \left[ u_1(s, t_{i^*(v)}) - P(s, t_{i^*(v)}) \right]\\ 
        =&\sum_{v \in A(r)} D(v) \left[ Z - (Z - r_{i^*(v)}) \right] \\
        =&\sum_{v \in A(r)} D(v) \cdot{} r_{i^*(v)}
    \end{align*}
    But this is exactly $\Rev(r)$ and so is at least $K$, as desired. 
    
    \underline{$\Leftarrow:$} Suppose there exists a payment function $P$ such that the leader's expected utility in $G[P]$ is at least $K$. 
    We'll assume without loss of generality that $P(s, t_i)\leq Z$ for all $i\neq 0$ and that $P(s, t_0) = 0$.\footnote{These assumptions are without loss of generality because replacing an arbitrary $P$ with one satisfying the assumptions only increases the leader's utility. 
    Note that if $P(s, t_i) > Z$ for any $i\neq 0$, the leader gets negative utility whenever $t_i$ is played, and similarly if $P(s, t_0) > 0$. 
    For any $P$, consider a $P'$ identical to $P$ except that $P'(s, t_i) = Z$ for any $i$ where $P(s, t_i) > Z$ and $P'(s, t_0) = 0$ if $P(s, t_0) > 0$.
    Switching to $P'$ from $P$ weakly increase the leader's utility because it only changes the best responses of follower types who would have played actions giving the leader negative utility. 
    These follower types give the leader nonnegative under $P'$ because all follower actions give the leader nonnegative utility under $P'$. 
    }
    Each follower type simply best responds to a fixed game $G[P]$, tiebreaking in favor of the leader, so we'll also assume without loss of generality that they play pure actions.\footnote{If the follower has multiple best responses after tiebreaking in favor of the leader, for consistency we'll assume they tiebreak in favor of the smallest index.
    }  

    We claim that the item pricing $r$ where $r_i = Z- P(s, t_i)$ yields revenue at least $K$ in $U$.
    This is a valid pricing in $U$ because $r_i = Z- P(s, t_i)\geq 0$ for all $i$.
    
    Let $j^*(v)$ be the index of the action 
    played by type $\theta_v$.
    For a payment commitment $P$, let $B(P)$ be the set of types $\theta_v$ which play an action aside from $t_0$, i.e. those with $j^*(v) > 0$.  

    Now, consider some $\theta_v \in B(P)$. 
    Type $\theta_v$'s utility for playing $j^*(v)$ is $u_2^{\theta_v}\left(s, t_{j^*(v)}\right) + P\left(s, t_{j^*(v)}\right) = -Z + v_{j^*(v)} + \left(Z - r_{j^*(v)}\right) = v_{j^*(v)} - r_{j^*(v)}$. 
    Since $\theta_v$'s utility for playing $t_0$ is $u_2^{\theta_v}(s, t_0) + P(s, t_0) = 0$, and $\theta_v$ instead plays $t_{j^*(v)}$, we have $v_{j^*(v)} - r_{j^*(v)} \geq 0$, and so by definition $v \in A(r)$. 
    In other words, for any follower type $\theta_v$ that plays a non-$t_0$ action in $G$, the corresponding buyer type with a value $v$ purchases an item in $U$. 
    
    Recall that $i^*(v)$ is defined by $i^*(v) \in \arg\max_{i>0} [v_{i} - r_{i}]$, tiebreaking in favor of large $r_i$.
    Observe that $\arg\max_{i>0} [v_{i} - r_{i}] = \arg\max_{i>0} [(v_i - Z) - (r_i - Z)] = \arg\max_{i>0} \left[u_2^{\theta_v}(s, t_{i}) + P(s, t_{i})\right]$. 
    Note that this operand $u_2^{\theta_v}(s, t_{i}) + P(s, t_{i})$ is follower type $\theta$'s utility for playing $t_i$ in $G[P]$.
    Furthermore, tiebreaking in favor of large $r_i$ is equivalent to tiebreaking in favor of the leader's utility in $G$ because the leader's utility when $t_j$ is played is $u_1(s, t_j) - P(s, t_j) = Z - P(s, t_j) = r_j$.   
    Therefore, $j^*(v)$ and $i^*(v)$ are identical for all $v$.\footnote{Assuming WLOG that we define both $j^*$ and $i^*$ to behave consistently (ex, to select the smallest index) when multiple indices achieve the maximum value of the $\arg\max$ and have the same tiebreaking value.}    
    In other words, for any follower type $\theta_v \in B(P)$ in $G$, the corresponding buyer type with value $v$ in $U$ will purchase item $j^*(v)$ and generate revenue $r_{j^*(v)}$. 
    Since the principal's revenue from buyers $v \not\in A(r)$ (who do not purchase items) is zero, we have
    \begin{align*}
        \Rev(R) &= \sum_{v \in A(r)} D(v) * r_{i^*(v)} \\ 
        &\geq \sum_{\theta_v \in B(P)} \Pr(\theta_v) * r_{j^*(v)} \\ 
        &= \sum_{\theta_v \in B(P)} \Pr(\theta_v) \left[ Z - (Z - r_{j^*(v)}) \right] \\
        &= \sum_{\theta_v \in B(P)} \Pr(\theta_v) \left[ u_1(s, t_j) - P(s, t_j) \right].
    \end{align*}    
    But this is exactly the leader's expected utility in $G$, and so is at least $K$, as desired.%
\end{proof}

\subsection{Leader Types}

\thmLeaderTypesPureCommitmentTwoPlayerHardness*

\begin{proof}
    We'll prove this via a reduction from \textsc{Vertex Cover}. 
    Consider a vertex cover instance $G = (V, E)$. 
    We'll construct a corresponding two-player Bayesian game $\Gamma$ with leader types only in which there exists a pure commitment for the leader achieving utility at least $k$ (for any $k\in(0, 1]$) if and only if there exists a $K$-cover in $G$, regardless of the leader's access to a correlation device.    
    As noted in the proof sketch, our reduction is the same as that of \citet{conitzer_optimal_strategy_to_commit} to the version of this setting without payments or signaling.

    We now define $\Gamma$ as follows. 
    The leader has $K$ types $\theta_1, \dots, \theta_K$, each of which occurs with probability $1/K$. 
    The leader has actions $a_v$ for each vertex $v \in V$.
    The follower has actions $b_e$ for each $e \in E$ and an additional action $b_0$. 

    The leader's utility function is the same for all types: $1$ if the follower plays $b_0$ and $0$ otherwise. 

    Each leader action corresponds to a vertex $v$, and intuitively, the follower is trying to play an edge $b_e$ which is not covered by the vertex $v$.
    The follower's utility function is: 
    \begin{itemize}
        \item $1$ for $b_e$ if $e$ is uncovered: $u_2(a_v, b_e) = 1$ if $v \not \in e$ for any $v\in V$ and $e\in E$   
        \item $-K$ for $b_e$ if $e$ is covered: $u_2(a_v, b_e) = 1$ if $v \in e$   for any $v\in V$ and $e\in E$   
        \item $0$ guaranteed if he plays $b_0$: $u_2(a_v, b_0) = 0$ for all $v \in V$
    \end{itemize}

    Consider some $k \in (0,1]$. 
    We claim that the leader can achieve expected utility at least $k$ via a pure commitment in $\Gamma$ if and only if there exists a $K$-cover in $G$, regardless of whether she has access to a signaling device. 
    That is, this biconditional holds both in the case that the leader's commitment includes a signaling scheme and in the case that it does not. 
    For the backwards direction, we show that the leader can achieve a utility of $k$ without access to a signaling device, which immediately implies that she can achieve the same with signaling. 
    For the forwards direction, we show that the leader cannot achieve a utility of $k$ even with access to a signaling device, which of course implies that she cannot achieve it without signaling either.

    \underline{$\Leftarrow$:} 
    Suppose there exists a $K$-cover $C \subseteq V$. 
    We demonstrate a leader commitment without signaling which achieves a utility of $1$. 
    Suppose the leader commits to the zero payment function and to actions such that, for each $v \in C$, a different one of her $K$ types plays the pure action $a_v$.

    Since $C$ is a cover, for every follower action $b_e$ (where $e = (v_1, v_2)$), at least one of $a_{v_1}$ and $a_{v_2}$ is played by some leader type $\theta_i$.
    Therefore, the follower's expected utility for playing any $b_e$ is at most $\frac{1}{K} (-K) + \frac{K-1}{K}(1) = -\frac{1}{K}$. 
    Hence, $b_0$ is the unique best response, and the leader receives utility $1$.
    
    \underline{$\Rightarrow$:}
    For a proof by contrapositive, suppose there does not exist a $K$-cover of $G$.
    Fix some incentive-compatible leader commitment to a payment function, pure action for each type, and signaling scheme. 
    Consider the leader's action profile conditional on sending recommendation $b_0$.
    There are at most $K$ distinct leader types for which $b_0$ is recommended with positive probability, and since each must play a pure action, there are at most $K$ distinct vertices $v$ for which $a_v$ is played with positive probability. 
    Since there is no $K$-cover, there must be some $e\in E$ which is not covered by any of the vertices $v$ for which $a_v$ is played. 
    The follower's utility for playing the corresponding $b_e$ is $1$. 
    Since $u_2(a_v, b_0) = 0$ for all $a_v$, the follower will only follow a recommendation to play $b_0$ if the leader commits to expected payments of at least $1$ when $b_0$ is played, in which case the leader's expected utility when recommending $b_0$ would be at most $1-1=0$.
    But the leader's utility when recommending $b_e$ for any $e \in E$ can be at most $0$ since $u(a_v, b_e)=0$ for all $a_v$ and all $b_e$.
    Therefore, the leader cannot achieve utility $k>0$ if there does not exist a $K$-cover of $G$. 
\end{proof}

\begin{lemma}\label{lemma:n=2_mixed_leader-types_payments-only-depend-on-follower-action}
    For a Bayesian game with $n=2$ players and a single follower type, any commitment to a (mixture over) actions and a payment function $P$ is equivalent (in terms of outcomes and expected payoffs) to one in which the payment function $P$ depends only on the follower's action. 
\end{lemma}
\begin{proof}
    Intuitively, the lemma holds because the follower's incentive to play an action $b$ depends only on her expected payment (over the leader's type and action) and not on any further details about the payment function. 
    
    Consider a commitment $(\sigma_1, P)$, in which $P$ is arbitrary. 
    The follower's expected utility for taking any action $b$ is 
    $$\sum_{\theta \in \Theta} \pi(\theta) \sum_{a \in A} \sigma_1^\theta(a) \left[ u_2(a, b) + P(a, b)\right].$$

    Now, consider an alternative payment function $P'$ which, whenever the follower plays some action $b$, pays them their expected payment under $P$, where the expectation is over the leader's type and action: 
    $P'(b) = \sum_{\theta \in \Theta} \pi(\theta) \sum_{a \in A} \sigma_1^{\theta}(a) P(a, b)$. 
    The follower's expected utility for any action $b$ is identical under $P$ and $P'$:
    \begin{align*}
        &\sum_{\theta \in \Theta} \pi(\theta) \sum_{a \in A} \sigma_1^\theta(a) \left[ u_2(a, b) + P(a, b)\right]\\
        =&\sum_{\theta \in \Theta} \pi(\theta) \sum_{a \in A} \sigma_1^\theta(a) u_2(a, b) +\sum_{\theta \in \Theta} \pi(\theta) \sum_{a \in A} \sigma_1^\theta(a) P(a, b) \\
        =& \sum_{\theta \in \Theta} \pi(\theta) \sum_{a \in A} \sigma_1^\theta(a) u_2(a, b) + P'(b)
    \end{align*}
    Therefore, the set of best responses are identical under $P$ and $P'$, and hence the expected payments and induced distributions over outcomes and are identical as well.
\end{proof}
A very similar argument shows that allowing the leader's payments to depend on her type also makes no difference to the problem. 

\thmTwoPlayerLeaderTypesMixedStrategiesPolytime*
\begin{proof}
    We extend our approach from the corresponding non-Bayesian setting of \Cref{Thm:2-player-mixed-strat}: We use linear programming to compute, for each follower action $b$, the leader's utility maximizing strategy that incentivizes $b$. 
    For the same reasons as before, there exists an optimal leader commitment for which the follower plays a pure action. 
    Hence, we can maximize over the $|B|$ follower actions to find an optimal leader's commitment and the corresponding utility. 
    Consider the following linear program. 

\begin{align*}
    &\text{Maximize} \quad \sum_{\theta \in \Theta} \pi(\theta) \sum_{a\in A} p_a^\theta  u_1^\theta(a, b) - P \\
    &\text{Subject to:} \\
    &P \geq 0 \\
    &\sum_{a\in A} p_a^\theta = 1 \hspace{12pt}\text{for all } \theta \in \Theta \\
    &p_a^\theta \geq 0  \hspace{30pt}\text{for all } \theta \in \Theta \text{ and } a \in A  \\
    &\sum_{\theta \in \Theta} \pi(\theta) \sum_{a \in A} p_a^\theta u_2(a, b) + P \geq \sum_{\theta \in \Theta} \pi(\theta) \sum_{a\in A} p_a^\theta u_2(a, b') \\ &\quad\quad \text{for all } b' \in B \setminus \{ b\}
\end{align*}

    The variables are $\{p_a^\theta\}$ and $P$, all of which are real valued scalars. 
    The utilities $u_i(\cdot)$, probabilities $\pi(\theta)$, and sets $\Theta$, $A$, and $B$ are parameters of the game instance, and so are constants from the perspective of the LP.
    The distinguished follower action $b$ is also an input to the LP. 
    It's easy to verify that the program is indeed linear in the variables. 

    The linear program is very similar to before. 
    We now have a probability distribution over leader actions for each leader type $\theta$ given by the variables $p_i^\theta$ and compute the leader's and follower's utilities in expectation over both the leader's type and the randomness in her mixture over actions. 
    By \Cref{lemma:n=2_mixed_leader-types_payments-only-depend-on-follower-action}, it suffices to consider payment functions that depend only on the follower's action. 
    Since the leader seeks to incentivize only action $b$, there is no benefit to making non-zero payments when the follower takes some $b'$. 
    Hence, the leader's choice of payment function amounts to choosing $P(b)$, which is what variable $P$ represents. 

    The objective is the leader's \textit{ex ante} expected utility after payments. 
    The first constraint ensures that the payment is non-negative.
    The next two sets of constraints ensure that, for each leader type $\theta$, $\{p_a^\theta\}$ represents a valid mixture over actions.
    The final set of constraints ensures that $b$ is a best response for the follower: the follower's expected utility for playing $b$ is at least that for playing any $b'$, where the expectations are over the leader's type and mixture over actions.  
    It is easy to see that these constraints are necessary and sufficient for the variables to correspond to a valid commitment incentivizing $b$.

    Observe that the LP has $O(|A|\cdot{}|\Theta|$) variables and $O(|A|\cdot{}|\Theta| + |B|)$ constraints. 
    Since these are both polynomial in the instance size, the LP can also be solved in polynomial time.
\end{proof}

\thmCELeaderTypesMixedAction*

\begin{proof}
    We extend our approach from \Cref{thm:correlated-n-player-mixed-normal-form} to handle leader types, again giving a polynomially sized linear program which directly computed the leader's optimal commitment. 
    Note that our LP has $O(|A|\cdot{}|\Theta|)$ variables and $O(|A|\cdot{}|\Theta|)$ constraints, each of which are polynomial in the instance size, so can be solved in polynomial time. 
    Hence, the correctness of the following LP implies the result. 
    
\begin{align*}
    &\text{Maximize} \quad \sum_{\theta \in \Theta} \pi(\theta) \sum_{a \in A} p_a^\theta u_1^\theta(a) - \sum_{i=2}^n t_i(a_i) \\
    &\text{Subject to:} \\
    &p_a^\theta \geq 0  \hspace{30pt}\text{for all } \theta \in \Theta \text{ and } i \in \{1, 2, \ldots, n\}  \\
    &\sum_{i = 1}^{|S|} p_a^\theta = 1 \hspace{12pt}\text{for all } \theta \in \Theta \\
    &t_i(a_i) \geq 0 \quad \text{for all } i \in \{2, \dots, n\} \text{ and } a_i \in A_i \\
    & \sum_{\theta \in \Theta} \pi(\theta) \sum_{a_{-i}} p_{(a_i, a_{-i})}^\theta u_i(a_i, a_{-i}) + t_i(a_i) \\ &\quad \geq \sum_{\theta \in \Theta} \pi(\theta) \sum_{a_{-i}} p_{(a_i, a_{-i})}^\theta u_i(a_i', a_{-i}) \\ &\quad \text{for all }  i \in \{2, \dots, n\} \text{ and } a_i \in A_i
\end{align*}

    The variables are the $p^\theta_a$ and $t_i(a_i)$, all of which are real valued scalars. 
    The $p^\theta_a$ represent the probability that action profile $a$ is recommended when the leader's realized type is $\theta$. 
    The $t_i(a_i)$ represent the payment from the leader to each follower $i$ for following a recommendation to play action $a_i$. 
    The utilities $u_i(\cdot)$ and $u_1^\theta(\cdot)$ are parameters of the game instance and so are constants from the perspective of the program; It's easy to verify that the program is linear.  

    The first two sets of constraints ensure that the $p^\theta_a$ variables form a valid probability distribution over action profiles for each leader type. 
    The third ensures that all payments are non-negative.
    Collectively, these first three sets of constraints ensure the variables correspond to a valid leader commitment. 
    The fourth set ensures that playing a recommended action $a_i$ is always a best response the follower's beliefs about the other players' actions conditional on the recommendation. 
    This means the leader's commitment is incentive compatible and hence that the objective correctly computes the leader's \textit{ex ante} expected utility. 
\end{proof}

\end{document}